\documentclass[11pt,letterpaper]{article}
\usepackage[margin=1.1in]{geometry}
\usepackage{amsmath,amssymb,amsfonts,setspace,url,enumitem}
\usepackage{amsthm,hyperref,mathrsfs}
\hypersetup{
    colorlinks=true,
    linkcolor=black,
    citecolor=blue
}
\usepackage{tikz}
\usepackage{color}
\usepackage{bigstrut}
\usepackage{multirow}
\usepackage{cleveref} %easy referencing via \Cref
\usepackage{ascmac}
\usepackage[thinlines]{easytable}
\usepackage{mathpazo}

\renewcommand{\arraystretch}{1.2}

\makeatletter
\newcommand{\newreptheorem}[2]{\newtheorem*{rep@#1}{\rep@title}\newenvironment{rep#1}[1]{\def\rep@title{#2 \ref*{##1}}\begin{rep@#1}}{\end{rep@#1}}}
\makeatother
\makeatletter
\newcommand{\newconjecture}[2]{\newconjecture*{rep@#1}{\rep@title}\newenvironment{rep#1}[1]{\def\rep@title{#2 \ref*{##1}}\begin{rep@#1}}{\end{rep@#1}}}
\makeatother
\theoremstyle{plain}
\newtheorem{theorem}{Theorem}
\newreptheorem{theorem}{Theorem}
\newreptheorem{conjecture}{Conjecture}

\newtheorem{proposition}{Proposition}

\theoremstyle{definition}

\NewDocumentEnvironment{mytheorem}{m}%
  {\renewcommand{\thetheorem}{$\mathbf{#1}$}%
   \begin{theorem}
  }%
  {\end{theorem}}

\newcommand{\set}[2]{\{#1,\ldots,#2\}}
\newcommand{\Int}{\mathbb{Z}}

\newcommand{\Real}{\mathbb{R}}
\newcommand{\Rat}{\mathbb{Q}}
\newcommand{\Nat}{\mathbb{N}}
\newcommand{\ent}{\mathcal{H}}
\providecommand{\Aa}{\mathcal{A}}
\providecommand{\Bb}{\mathcal{B}}
\providecommand{\Cc}{\mathcal{C}}
\providecommand{\Dd}{\mathcal{D}}
\providecommand{\Ccc}{\mathscr{C}}

\providecommand{\Tt}{\mathcal{T}}
\providecommand{\Nn}{\mathcal{N}}
\providecommand{\Mm}{\mathcal{M}}

\newcommand{\ceil}[1]{\left\lceil #1 \right\rceil}

\providecommand{\myb}{\lambda_b}

\newcommand{\tcw}{T_{\mathsf CW}}

\newcommand{\field}{\mathbb{F}}
\newcommand{\braket}[1]{\left\langle #1 \right\rangle}
\newcommand{\degen}{\unlhd}

\begin{document}

\title{Faster Rectangular Matrix Multiplication\\ by Combination Loss Analysis}
\author{
Fran{\c c}ois Le Gall\\
Nagoya University\\
legall@math.nagoya-u.ac.jp
}

\date{}
\maketitle
\thispagestyle{empty}
\setcounter{page}{1}
\begin{abstract}
Duan, Wu and Zhou (FOCS 2023) recently obtained the improved upper bound on the exponent of square matrix multiplication $\omega<2.3719$ by introducing a new approach to quantify and compensate the ``combination loss" in prior analyses of powers of the Coppersmith-Winograd tensor.
  In this paper we show how to use this new approach to improve the exponent of rectangular matrix multiplication as well. Our main technical contribution is showing how to combine this analysis of the combination loss and the analysis of the fourth power of the Coppersmith-Winograd tensor in the context of rectangular matrix multiplication developed by Le Gall and Urrutia (SODA 2018).
\end{abstract}

%\newpage
%=====================================
\section{Introduction}\label{sec:intro}
%=====================================
\subsection{Prior works on the exponent of matrix multiplication}
\paragraph{Square matrix multiplication.}
Matrix multiplication is one of the most fundamental computational tasks. The exponent of square matrix multiplication (denoted $\omega$), in particular, is a central and ubiquitous quantity in theoretical computer science. The exponent of (square) matrix multiplication represents the exponent of the asymptotic complexity of the best possible matrix multiplication algorithm: it can be defined as the smallest $\omega$ such that two $n\times n$ matrices can be multiplied in $O(n^{\omega+\epsilon})$ time for any $\epsilon>0$. 
It is easy to show that $\omega\in[2,3]$, but the precise value of $\omega$ is still unknown. The first non-trivial upper bound $\omega<2.81$ was obtained in 1969 by Strassen \cite{Strassen69}, and later improved several times  \cite{Bini+79,Coppersmith+82,Pan79,Pan81,Romani82, Schonhage81,StrassenFOCS86}. In particular, Strassen obtained in 1986 the upper bound $\omega<2.48$ by developing a new technique called  the laser method \cite{StrassenFOCS86}.

Since 1987, all new upper bounds on $\omega$ have been obtained by applying the laser method to a mathematical construction called the Coppersmith-Winograd tensor \cite{Coppersmith+90}, which we denote $\tcw$ in this paper. First, Coppersmith and Winograd \cite{Coppersmith+90} analyzed $\tcw$ and its second power $\tcw^{\otimes 2}$ using the laser method and obtained the bound $\omega<2.3754770$. More than twenty years later, Stothers~\cite{Stothers10} and Vassilevska Williams \cite{WilliamsSTOC12} were able to analyze the fourth power $\tcw^{\otimes 4}$ and obtained the new upper bound $\omega<2.3729269$ (see also \cite{Davie+13,LeGall14}). The approach by Vassilevska Williams was especially powerful since it made possible to analyze recursively powers of the Coppersmith-Winograd tensor for $m=2^\ell$ with any $\ell\ge 1$. Analyzing the eighth power using this approach, Vassilevska Williams \cite{WilliamsSTOC12} obtained the improved bound $\omega<2.3728642$. Deriving upper bounds on $\omega$ using this approach nevertheless requires solving a complicated (in particular, non-convex) optimization problem, which becomes extremely challenging for $\ell>3$. Le Gall~\cite{LeGall14} showed that this optimization problem can be relaxed into a convex optimization problem, which made possible to completely analyze $\tcw^{\otimes m}$ for $m=16$ and $m=32$, and consequently obtained the upper bound $\omega<2.3728639$. Recently, Alman and Vassilevska Williams~\cite{Alman+SODA21} showed how to refine the analysis of one key steps of the laser method, and consequently obtained an improved bound $\omega<2.3728596$.  All these bounds are reported in Table~\ref{table:chart}. We also mention a series of works \cite{Alman2021,Alman+ITCS18,Alman+FOCS18,Ambainis+STOC15,Christandl+20,Christandl+21} showing the limits of these approaches (in particular, the impossibility to prove $\omega=2$ from the Coppersmith-Winograd tensor).

Very recently, Duan, Wu and Zhou \cite{Duan+23}  obtained the improved upper bound $\omega<2.371866$. Their key discovery is that prior analyses of powers of the Coppersmith-Winograd tensor suffer from a  ``combination loss''.  Duan, Wu and Zhou \cite{Duan+23} gave a quantitative analysis of this combination loss, showed how to compensate it, and applied this methodology to the fourth power and the eighth power of the Coppersmith-Winograd tensor. This new approach crucially requires an asymmetric analysis of the powers of the Coppersmith-Winograd tensor  (prior works on square matrix multiplication only needed a symmetric analysis). \vspace{-2mm}

\begin{table}[tb!]
\renewcommand\arraystretch{1}% (1.0 is for standard spacing)
\begin{center}
\begin{tabular}{|l|c|l|l|}
%\begin{tabular}{ l l l l}
\hline
Upper bound&$m$ & Reference\bigstrut&Technique\\ 
\hline
$\omega<2.3871900$ &1& Coppersmith and Winograd~\cite{Coppersmith+90}&Laser method\\
\hline
$\omega<2.3754770$ &2&   Coppersmith and Winograd~\cite{Coppersmith+90}&Laser method\\
\hline
\multirow{2}{*}{$\omega<2.3729269$}&\multirow{2}{*}{4}&  Stothers \cite{Stothers10}, Vassilevska Williams \cite{WilliamsSTOC12} &\multirow{2}{*}{Recursive laser method}\\
&& (see also \cite{Davie+13,LeGall14})  &\\
\hline
$\omega<2.3728642$&8&  Vassilevska Williams \cite{WilliamsSTOC12}&Recursive laser method\\
\hline
$\omega<2.3728640$&16&  Le Gall \cite{LeGall14}&Recursive laser method\\
\hline
$\omega<2.3728639$&32&  Le Gall \cite{LeGall14}&Recursive laser method\\
\hline
$\omega<2.3728596$&32& Alman and Vassilevska Williams \cite{Alman+SODA21}&Refined laser method\\
\hline
\hline
$\omega<2.371919$ &4&  Duan, Wu and Zhou \cite{Duan+23}&Combination loss analysis\!\\
\hline
$\omega<2.371866$ &8&  Duan, Wu and Zhou \cite{Duan+23}&Combination loss analysis\!\\
%\hline
%$\omega<2.370677$ &4&  This paper&Refined asymmetric analysis
%\tabularnewline
\hline
\end{tabular}
\caption{Upper bounds on $\omega$ obtained by analyzing the $m$-th power of the Coppersmith-Winograd tensor.}\label{table:chart}\vspace{-3mm}
\end{center}
\end{table}

\paragraph{Rectangular matrix multiplication.}
Rectangular matrix multiplication appears as a bottleneck in several computational problems (e.g., the construction of fast algorithms for the all-pairs shortest paths problem \cite{Alon+ESA07, Roditty+11,YusterSODA09,ZwickSTOC99,ZwickJACM02},
the dynamic computation 
of the transitive closure \cite{Demetrescu+FOCS00,Sankowski+10},
finding ancestors \cite{Czumaj+TCS07} or detecting directed cycles \cite{Yuster+SODA04}).
From a theoretical perspective, the most relevant quantity is the exponent of rectangular matrix multiplication: for any $\kappa\ge 0$, the exponent of rectangular matrix multiplication $\omega(\kappa)$ is defined as the smallest $\omega(\kappa)$ such that the product of an $n\times \ceil{n^\kappa}$ matrix by an $\ceil{n^\kappa}\times n$ matrix can be computed in $O(n^{\omega+\epsilon})$ time for any $\epsilon>0$.\footnote{It is known (see, e.g., \cite{Blaser13, Burgisser+97}) that the arithmetic complexity of the following three types of matrix products is the same: computing the product of an $n\times n$ matrix by an $n\times m$ matrix; computing the product of an $n\times m$ matrix by an $m\times n$ matrix; computing the product of an $m\times n$ matrix by an $n\times n$ matrix. The exponent of rectangular matrix multiplication thus represents the exponent of the asymptotic complexity of the best possible algorithm for any of these three kinds of matrix multiplication.} Note that $\omega=\omega(1)$.

There is a long history of research  on proving upper bounds on $\omega(\kappa)$ for $\kappa\neq 1$ as well \cite{CoppersmithSICOMP82,Coppersmith97,Huang+98,Ke+08,LeGallFOCS12,LeGall+SODA18,Lotti+83}.
The best known upper bounds, which are shown in Table \ref{table_SODA18}, have been obtained by Le Gall and Urrutia~\cite{LeGall+SODA18}. These upper bounds were obtained by analyzing the fourth power of the Coppersmith-Winograd tensor, and improved prior bounds obtained by analyzing the second power \cite{LeGallFOCS12} and the first power \cite{Coppersmith97}. 

\begin{table}[!tbp]
\begin{center}
\begin{tabular}{ |c | c |}
  \hline
  \multirow{2}{*}{$\kappa$} & upper bound \\%& \multirow{2}{*}{$q$}\\
  &on $\omega(\kappa)$\\%&\\
  \hline                      
0.31389 & 2\\
%0.31477 & 2.000001\\
0.32 & 2.000064\\
0.33 & 2.000448\\
0.34 & 2.001118\\
0.35 & 2.001957\\
0.40 & 2.010314\\
\hline  
\end{tabular}
\hspace{0.4cm}
\begin{tabular}{ |c | c |}
  \hline
  \multirow{2}{*}{$\kappa$} & upper bound \\%& \multirow{2}{*}{$q$}\\
  &on $\omega(\kappa)$\\%&\\
  \hline
0.50 & 2.044183\\
0.60 & 2.093981\\
0.70 & 2.154399\\
0.80 & 2.222256\\
0.90 & 2.295544\\
1.00 & 2.372927\\
  \hline  
\end{tabular}
\hspace{0.4cm}
\begin{tabular}{ |c | c |}
  \hline
  \multirow{2}{*}{$\kappa$} & upper bound \\%& \multirow{2}{*}{$q$}\\
  &on $\omega(\kappa)$\\%&\\
  \hline
1.10 & 2.453481\\
1.20 & 2.536550\\
1.50 & 2.796537\\
2.00 & 3.251640\\
3.00 & 4.199712\\
5.00 & 6.157233\\
  \hline  
\end{tabular}
\end{center}\vspace{-3mm}
\caption{The upper bounds on $\omega(\kappa)$ from \cite{LeGall+SODA18}.
\label{table_SODA18}}
\end{table}

For $\kappa=1$ (square matrix multiplication), the results from \cite{LeGall+SODA18} recover the upper bound $\omega<2.3729269$ obtained from the analysis of fourth power of the Coppersmith-Winograd tensor by Stothers \cite{Stothers10} and Vassilevska Williams \cite{WilliamsSTOC12}. Note that higher powers of the Coppersmith-Winograd tensor (e.g., the eighth power) have not yet been analyzed in the context of rectangular matrix multiplication, mainly because the optimization problems are significantly more complicated than for square matrix multiplication. Ref.~\cite{LeGall+SODA18} also shows that $\omega(0.31389)=2$ and thus gives the lower bound $0.31389$ on the quantity $\sup\{\kappa\:|\:\omega(\kappa)=2\}$ that is called the dual exponent of matrix multiplication (and sometimes denoted $\alpha$).

\subsection{Statement of our results}
%\paragraph{Our results.}
%The starting point of this work is the observation that two prior works \cite{LeGallFOCS12,LeGall+SODA18} have already developed techniques for the asymmetric analysis of some powers of the Coppersmith-Winograd tensor, in the context of \emph{rectangular} matrix multiplication: Le Gall \cite{LeGallFOCS12} analyzed the second power of the Coppersmith-Winograd tensor, and Le Gall and Urrutia \cite{LeGall+SODA18} then generalized that approach to the fourth power. Part of these analyses and the optimization techniques are actually slightly more general that the asymmetric analysis used in \cite{Duan+23}.

In this work, we show how to improve the results from \cite{LeGall+SODA18} by applying the recent approach from \cite{Duan+23} to analyze the combination loss and consequently refine the analysis of the fourth power of the Coppersmith-Winograd tensor in the context of rectangular matrix multiplication as well. Our new upper bounds on $\omega(\kappa)$ are shown in Table \ref{table_ours} for the same values of $\kappa$ as in Table~\ref{table_SODA18}. While we do not currently obtain a better lower bound on the dual exponent of matrix multiplication, our upper bounds on $\omega(\kappa)$ improve the bounds from \cite{LeGall+SODA18} for all $\kappa>0.31389$.\vspace{2mm}

\begin{table}[!htbp]
\begin{center}
\begin{tabular}{ |c | c |}
  \hline
  \multirow{2}{*}{$\kappa$} & upper bound \\%& \multirow{2}{*}{$q$}\\
  &on $\omega(\kappa)$\\%&\\
  \hline                      
0.31389 & 2\\
%0.31477 & 2.000001\\
0.32 & 2.000059\\
0.33 & 2.000355\\
0.34 & 2.000894\\
0.35 & 2.001726\\
0.40 & 2.010118\\
\hline  
\end{tabular}
\hspace{0.4cm}
\begin{tabular}{ |c | c |}
  \hline
  \multirow{2}{*}{$\kappa$} & upper bound \\%& \multirow{2}{*}{$q$}\\
  &on $\omega(\kappa)$\\%&\\
  \hline
0.50 & 2.044076\\
0.60 & 2.093897\\
0.70 & 2.154283\\
0.80 & 2.222075\\
0.90 & 2.295254\\
1.00 & 2.372537\\
  \hline  
\end{tabular}
\hspace{0.4cm}
\begin{tabular}{ |c | c |}
  \hline
  \multirow{2}{*}{$\kappa$} & upper bound \\%& \multirow{2}{*}{$q$}\\
  &on $\omega(\kappa)$\\%&\\
  \hline
1.10 & 2.452999\\
1.20 & 2.535921\\
1.50 & 2.795600\\
2.00 & 3.250563\\
3.00 & 4.199095\\
5.00 & 6.156708\\
  \hline  
\end{tabular}
\end{center}\vspace{-3mm}
\caption{Our new upper bounds on $\omega(\kappa)$.
\label{table_ours}}
\end{table}

Note that for $\kappa = 1$ (square matrix multiplication) we obtain the upper bound $\omega<2.372537$, which is weaker than the upper bound $\omega<2.371919$ from \cite{Duan+23} obtained by analyzing the fourth power of the Coppersmith-Winograd tensor in the context of square matrix multiplication. This is because our framework is specific to rectangular matrix multiplication, and several refined optimization steps from \cite{Duan+23} are difficult to implement in the context of rectangular matrix multiplication.\footnote{The difficulties are both theoretical and practical. From the theoretical perspective, for rectangular matrix multiplication we cannot use the notion of ``value'' of a tensor, which makes the analysis more difficult. From the practical perspective, the optimization problem that arises when considering the fourth power of the Coppersmith-Winograd tensor in the context of rectangular matrix multiplication is significantly more difficult to solve, and solving it after adding the most refined optimization steps from \cite{Duan+23} seems extremely challenging.} 

\subsection{Technical overview of the paper}
We now give an overview of both the approach by Duan, Wu and Zhou \cite{Duan+23} to analyze the ``combination loss'' (in the context of square matrix multiplication) and the approach by Le Gall and Urrutia \cite{LeGall+SODA18} to perform asymmetric analysis of the fourth power of the Coppersmith-Winograd tensor (in the context of rectangular matrix multiplication), and explain how to combine both approaches to get our improved upper bound on $\omega(\kappa)$. In particular, we introduce six fundamental conditions (\eqref{eq:C1}, \eqref{eq:C2}, \eqref{eq:C4}, \eqref{eq:CC4}, \eqref{eq:D3}, \eqref{eq:DD3}) on the parameters (we use the same labels for these conditions as in Section \ref{sec:asym}  --- the other conditions \eqref{eq:C3}, \eqref{eq:D1}, \eqref{eq:D2}, \eqref{eq:D4}, \eqref{eq:E1}, \eqref{eq:E2} will be introduced in Section \ref{sec:asym}).

\paragraph{The fourth power of the Coppersmith-Winograd tensor.}
The fourth power of the Coppersmith-Winograd tensor can be decomposed as follows:
\[
T_{\mathsf CW}^{\otimes 4}=\sum_{(ijk)\in S_8}T_{ijk}
\]
where $S_8=\left\{(ijk)\in\{0,\ldots,8\}^3\:\vert\: i+j+k=8\right\}$ and each $T_{ijk}$ is a ``smaller'' tensor called a component. The laser method analyzes this tensor by  assigning a weight $\alpha_{ijk}\in[0,1]$ to each component~$T_{ijk}$ The assignment satisfies the condition
\begin{equation}\tag{C1}
\sum_{(i,j,k)\in S_8}\alpha_{ijk}=1 
\end{equation}
and thus the set $\{\alpha_{ijk}\}$ can be considered as a probability distribution, which we denote below by $\alpha$. In the symmetric analysis used in \cite{Alman+SODA21,Davie+13,LeGall14,Stothers10,WilliamsSTOC12}, the distribution $\alpha$ is chosen completely symmetric, i.e., $\alpha_{ijk}=\alpha_{ikj}=\alpha_{jik}=\alpha_{jki}=\alpha_{kij}=\alpha_{kji}$ for all $(ijk)\in S_8$. The asymmetric analysis from \cite{Duan+23,LeGall+SODA18}, on the other hand, only imposes the condition\footnote{While Ref.~\cite{LeGall+SODA18} imposes Condition \eqref{eq:C2}, Ref.~\cite{Duan+23} actually imposes the condition $\alpha_{ijk}=\alpha_{jik}$ (i.e., symmetry of the first and second indices). In this work we adopt the former symmetry condition.}
\begin{equation}\tag{C2}
\alpha_{ijk}=\alpha_{ikj} \:\:\textrm{ for all }\:\: (ijk)\in S_8.
\end{equation}

\paragraph{Standard laser method and symmetric analysis.}
As mentioned above, $\alpha$ is a probability distribution over $S_8$. We denote by $\Aa,\Bb,\Cc\colon\set{0}{8}\to[0,1]$ the marginal distributions of $i$, $j$ and $k$, respectively. From the symmetry condition \eqref{eq:C2}, we have $\Bb=\Cc$. For square matrix multiplication, the potential of a tensor to give a good upper bound on $\omega$ is quantified by the concept of ``value'' (a higher value gives a better upper bound on $\omega$). The standard laser method shows that the ``value'' of the tensor $T_{\mathsf CW}^{\otimes 4}$ is lower bounded by the quantity
\[
\min\{\ent(\Aa),\ent(\Bb)\} + \Mm,
\]
or, (almost) equivalently, lower bounded by the quantity 
\begin{equation}\label{eq:val}
\ent(\Bb) + \Mm
\end{equation}
under the condition 
\begin{equation}\tag{C4}
\ent(\Aa)\ge\ent(\Bb).
\end{equation}
Here $\Mm$ quantifies the contribution of the $T_{ijk}$'s (weighted by the $\alpha_{ijk}$'s) to the ``value" of the whole tensor $T_{\mathsf CW}^{\otimes 4}$. 
%This term depends on new parameters: for each $(ijk)\in S_8$ we have parameters $g_{ijk\ell}$ for some $\ell\le 4$. We denote the whole set of parameters by $g$. 

Under Condition \eqref{eq:C4}, the quantity \eqref{eq:val} is optimized for a fully symmetric distribution, for which we have $\ent(\Aa)=\ent(\Bb)$. This is why prior works on square matrix multiplication before~\cite{Duan+23} used a symmetric analysis, i.e., only considered a fully symmetric distribution $\alpha$.

%Working out the details of the analysis and then finding the parameters that give the largest value, which will actually be done explicitly in Section \ref{sec:asym} (see in particular Section \ref{sub7}), gives the bound 
%\[
%\omega<2.3729269,
%\]
%which is exactly the bound obtained in \cite{WilliamsSTOC12} (see also \cite{LeGall14}) for the analysis of the fourth power of the Coppersmith-Winograd tensor. 

\paragraph{The approach by Duan, Wu and Zhou for exploiting the combination loss.}
The approach in \cite{Duan+23} refined the above analysis, and showed that the ``value'' of the tensor $T_{\mathsf CW}^{\otimes 4}$ is actually lower bounded by the quantity
\[
\min\{\ent(\Aa)-\chi,\ent(\Bb)\} + \Mm,
\]
or, (almost) equivalenty, lower bounded by the quantity 
\[
\ent(\Bb) + \Mm
\]
under the condition 
\begin{equation}\tag{C4'}
\ent(\Aa)-\chi\ge\ent(\Bb).
\end{equation}
Here, $\chi$ is a complicated quantity (depending on $\alpha$ and other parameters) that represents the ``combination loss" in the analysis by the standard laser method. Crucially, we have $\chi\le 0$, and thus Condition \eqref{eq:CC4} is a relaxation of Condition \eqref{eq:C4}.
%Choosing parameters such that $\ent(\Aa)>\ent(\Bb)$ can now be useful to optimize the quantity \eqref{eq:val2}. 
Note that having parameters satisfying \eqref{eq:CC4} but not \eqref{eq:C4}, i.e., parameters such that $\ent(\Bb)>\ent(\Aa)$, can only happen in the asymmetric case, which is why the analysis of~\cite{Duan+23} has to be asymmetric to lead to an improvement on $\omega$.
%We can get a larger value, and thus a better upper bound on $\omega$ if $p(\alpha)<0$ and $\ent(\Aa)>\ent(\Bb)$. For the symmetric analysis we have

\paragraph{Asymmetric analysis by Le Gall and Urrutia for rectangular matrix multiplication.}
While \cite{LeGall+SODA18} did not consider the combination loss, it also required an asymmetric analysis of the fourth power of the Coppersmith-Winograd tensor in order to obtain good bounds on $\omega(\kappa)$ for $\kappa\neq 1$.

For deriving bounds on $\omega(\kappa)$ for $\kappa\neq 1$, the notion of ``value'' cannot be used anymore. Instead, the goal is to show that the tensor $T_{\mathsf CW}^{\otimes 4}$ can be ``converted'' into $r$ copies of a tensor corresponding to the product of an $m\times m$ matrix by an $m\times m^\kappa$ for some $r$ and $m$ depending on the parameters (in particular, depending on the distribution $\alpha$). This gives, via the ``asymptotic sum inequality'' (see Proposition \ref{th:value} in Section \ref{sec:prelim}), the bound
\[
r\cdot m^{\omega(\kappa)}\le (q+2)^4,
\]
and thus an upper bound on $\omega(\kappa)$. The goal is thus to optimize the terms $r$ and $m$ in the conversion in order to obtain the best possible upper bound on $\omega(\kappa)$.

The term $r$ can be decomposed in two parts. The first part can be analyzed similarly to the square case (see below). The second part is much more difficult to analyze (it requires in particular a new condition, Condition \eqref{eq:E1} defined in Section \ref{sub7}).

In order to analyze the term $m$, the approach from \cite{LeGall+SODA18} showed how to perform a global analysis of the $T_{ijk}$'s with $(ijk)\in \bar S_8$, where
$
\bar S_8=\left\{(ijk)\in S_8\:\vert\: i,j,k>0\right\}.
$
 For each $(ijk)\in \bar S_8$ the approach introduced marginals distributions $\Aa_{ijk}$ and $\Bb_{ijk}$ (similar to the marginal $\Aa$ and~$\Bb$ defined above), and showed that the contribution of all these $T_{ijk}$'s can be analyzed globally under the condition 
\begin{equation}\tag{D3}
\sum_{(ijk)\in  \bar S_8}\alpha_{ijk}\ent(\Aa_{ijk}) \ge
\sum_{(ijk)\in \bar S_8}\alpha_{ijk}\ent(\Bb_{ijk}).
\end{equation}

\paragraph{Improving the analysis by exploiting the combination loss (first step).}
As already mentioned, the analysis of the first part of $r$ performed in \cite{LeGall+SODA18} 
is similar to the case of square matrix multiplication, and in particular depends on the quantity $\min\{\ent(\Aa),\ent(\Bb)\}$. Ref.~\cite{LeGall+SODA18} imposed Condition \eqref{eq:C4}, which implies $\min\{\ent(\Aa),\ent(\Bb)\}=\ent(\Bb)$, and then used $\ent(\Bb)$ in the analysis. Combined with the analysis of the second part of $r$ and $m$, this gives the bounds on $\omega(k)$ reported in Table \ref{table_SODA18}.
Concretely, for $\kappa=2$, this gives the upper bound 
\[
\omega(2)<3.251640.
\] 

Our first observation is that the analysis of the combination loss from \cite[section 6]{Duan+23} applies to the rectangular case as well. In Section \ref{sec:loss}, we show that this enables us to replace Condition \eqref{eq:C4} by the relaxed condition \eqref{eq:CC4}. Making this change already gives improved upper bounds on $\omega(\kappa)$. For instance, for $\kappa=2$, we obtain the improved bound
\[
\omega(2)<3.251502.
\]

\paragraph{Exploiting the combination loss recursively.}
The most general approach described in \cite[Sections 7 and 8]{Duan+23} actually applies the analysis of the combination loss recursively. Concretely, for the fourth power of the Coppersmith-Winograd tensor, this means that the analysis of the combination loss can be used to improve the analysis of the contribution of each $T_{ijk}$ as well. In Section \ref{sec:rec}, we implement this strategy in the context of rectangular matrix multiplication. This makes possible to replace Condition \eqref{eq:D3} by the condition
\begin{equation}\tag{D3'}
\sum_{(ijk)\in  \bar S_8}\alpha_{ijk}\left(\ent(\Aa_{ijk})-\chi_{ijk}\right)\ge 
\sum_{(ijk)\in \bar S_8}\alpha_{ijk}\ent(\Bb_{ijk}),
\end{equation}
where $\chi_{ijk}$ is a term that represents the combination loss occuring in the analysis of the term~$T_{ijk}$. Since $\chi_{ijk}\le0$ for each $(ijk)\in \bar S_8$, Condition \eqref{eq:DD3} is a relaxation of Condition \eqref{eq:D3}. With this additional relaxed condition, we obtain the upper bounds on $\omega(\kappa)$ shown in Table \ref{table_ours}.  For instance, for $\kappa=2$, we obtain the further improved bound
\[
\omega(2)<3.250563.
\]

\paragraph{Details about the optimization.}
Besides the theoretical analysis outlined above, a non-trivial contribution of this work is solving the corresponding optimization problem (which is necessary to find the set of parameters that gives the new upper bound on $\omega$). This optimization problem has 90 variables and 35 constraints, including 16 nonlinear constraints. In particular, imposing the global constraints \eqref{eq:D3}, \eqref{eq:DD3}, \eqref{eq:E1}, \eqref{eq:E2} makes the optimization problem significantly harder to solve than the optimization problem for the fourth power of the Coppersmith-Winograd tensor in~\cite{Duan+23}. Additionally, it is unclear how to apply the strategy for optimization used in~\cite{Alman+SODA21,Duan+23,LeGall14}, which consists in converting the problem into a convex optimizing problem and solving it using software for convex optimization, to the rectangular setting.
In consequence, we solve directly the original (nonconvex) optimization problem using the NLPSolve function in Maple \cite{Maple22}.\footnote{Let us mention one technical aspect of the optimization. To be able to solve this problem directly, we force the probability distribution $\alpha$ and probability distributions defined by $h$ (the distributions used to define the marginals $\Aa_{ijk}$ and $\Bb_{ijk}$) to have maximum entropy among all distributions with the same marginals (as was done in \cite{LeGallFOCS12,LeGall+SODA18,WilliamsSTOC12}). Concretely, this is implemented by imposing Conditions \eqref{eq:C3} and \eqref{eq:D4} in Section \ref{sec:asym}.} Details are given in Sections \ref{sub7}, \ref{sub:analysis-loss} and \ref{sub:rec-analysis}.
%(which was also used in \cite{Alman+SODA21,LeGall+SODA18,WilliamsSTOC12})

\paragraph{Remark about the presentation of the paper.} While our approach can be generalized to analyze higher powers of the Coppersmith-Winodrad tensor (e.g., $\tcw^{\otimes 8}$), in this paper we focus entirely on the fourth power. This enables us to use lighter notations and give closed-form expressions for many quantities, which (in our opinion) makes the paper significantly easier to read.\footnote{Another reason why we avoid generalizing to higher powers is that even if the theoretical framework can be derived, it would be extremely challenging to solve the resulting optimization problem, even for the eight power.}

The basis of this paper is the asymmetric analysis of the fourth power of the Coppersmith-Winograph by \cite{LeGall+SODA18}, which is presented in detail in Section \ref{sec:asym}. In Section \ref{sec:loss} we show how to modify this analysis to take in consideration the analysis of the combination loss from \cite[Section 6]{Duan+23}. In Section \ref{sec:rec}, we show how to further modify this analysis to take in consideration the analysis of the combination loss of the components from \cite[Section 7]{Duan+23}. 

%%%%%%%%%%%%%%%%%%%%%%%%%%%%%%%%%%%%%%%%%
\section{Preliminaries}\label{sec:prelim}
%%%%%%%%%%%%%%%%%%%%%%%%%%%%%%%%%%%%%%%%%
In this section we present some notations used in this paper (Section \ref{sub:prelim-gen}), explain the notions of algebraic complexity theory (Section \ref{sub:prelim-alg}), and introduce the Coppersmith-Winograd tensor and its second tensor power (Section \ref{sub:CW}).

%%%%%
\subsection{General notations}\label{sub:prelim-gen}
%%%%%
We use $\log(\cdot)$ to denote the binary logarithm.  Given a probability distribution $p\colon X\to[0,1]$ over a finite set $X$, we write
\[
\ent(p)=-\sum_{x\in X}^n p(x) \log(p(x)) 
\]
its entropy. We will often consider distributions over the set of integers $X=\set{0}{n}$, for some integer $n$. In this case we often write the distribution as $p=(p(1),p(2),\ldots,p(n))$.

We use $\Rat[0,1]$ to denote the set of rational numbers between 0 and 1.
%The following notation will be convenient in Section~\ref{sec:asym} to state our lower bounds on the value. 
%For two functions $f(N)$ and $g(N)$, we write $f(N)\gtrsim g(N)$ if $f(N)\ge g(N)-o(N)$.
For conciseness, when considering triples $(i,j,k)\in \Int\times\Int\times \Int$
we will often write $(ijk)$ instead of $(i,j,k)$. 

%%%%%
\subsection{Algebraic complexity theory}\label{sub:prelim-alg}
%%%%%
This subsection presents the notions of algebraic complexity needed for this work.
We refer to, e.g.,  \cite{Blaser13, Burgisser+97} for a more detailed treatment.
In this subsection $\field$ denotes an arbitrary field. 

%%%%%%%%%%%%%%%%%%%%%%%%
\paragraph{Tensors.}
%%%%%%%%%%%%%%%%%%%%%%%%
Let $U$, $V$ and~$W$ be three finite-dimensional vector spaces over
$\field$. A tensor (also called a trilinear form)~$t$ on $(U,V,W)$ is an element in $U\otimes V\otimes W$.
If we fix bases $\{x_i\}$, $\{y_j\}$ and $\{z_k\}$ of $U$, $V$ and $W$, respectively,
then $t$ can be written as
\[
t=\sum_{i,j,k}\gamma_{ijk}\:x_i\otimes y_j\otimes z_k
\]
for coefficients $\gamma_{ijk}$ in $\field$. We call $\{x_i\}$, $\{y_j\}$ and $\{z_k\}$ the $x$-variables, $y$-variables, and $z$-variables, respectively.

Matrix multiplication of an $m\times n$ matrix with entries in~$\field$ by an $n\times p$ matrix
with entries in~$\field$
corresponds to the following tensor on $(\field^{m\times n},\field^{n\times p},\field^{m\times p})$:
\[
\sum_{r=1}^m\sum_{s=1}^n\sum_{t=1}^p x_{rs}\otimes y_{st}\otimes z_{rt}.
\]
%This tensor will be denoted by $\braket{m,n,p}$. Another important example is the form 
%\[
%\sum_{\ell=1}^n x_\ell y_\ell z_\ell.
%\] 
%This tensor on $(\field^n,\field^n,\field^n)$ is denoted~$\braket{n}$ and corresponds to $n$ independent scalar products.

%Given a tensor $t\in U\otimes V\otimes W$, it will be convenient to 
%denote by $t_{\mathsf C}$ and $t_{\mathsf C^2}$ the tensors in $V\otimes W\otimes U$
%and $W\otimes U\otimes V$, respectively, obtained by permuting cyclicly the coordinates
%of $t$:
%\[
%t_{\mathsf C}=\sum_{ijk}t_{ijk}\:y_j\otimes z_k\otimes x_i,\hspace{5mm}t_{\mathsf C^2}=\sum_{ijk}t_{ijk}\:z_k\otimes x_i\otimes y_j.
%\]

Given two tensors
$t\in U\otimes V\otimes W$ and $t'\in U'\otimes V'\otimes W'$, their direct sum $t\oplus t'$ is a tensor in 
$(U\oplus U')\otimes (V\oplus V')\otimes (W\oplus W')$, and their tensor product is a tensor in 
$(U\otimes U')\otimes (V\otimes V')\otimes (W\otimes W')$. 
For any integer $e\ge 1$, the tensor $t\oplus\cdots\oplus t$ (with~$e$ occurrences of $t$) 
will be denoted
by $e\cdot t$ and 
the tensor $t\otimes\cdots \otimes t$ (with $e$ occurrences of~$t$)
will be denoted by $t^{\otimes e}$.

\paragraph{Degeneration, combinatorial restriction and border rank.}
The most general way to convert a tensor to another is via the concept of \emph{degeneration}. In this paper, we write $t'\degen t$ to denote that $t'$ is a degeneration of $t$.  Since we will almost never use this general concept (we only use it below to give the formal definition of the concepts border rank), we do not give the definition and instead refer the interested reader to, e.g., \cite{Blaser13}. 

In this paper, we will use a special kind of degeneration called combinatorial restriction (also called zeroing out), which has also been used in most recent works on matrix multiplications based on the laser method \cite{Alman+SODA21,Coppersmith+90,Davie+13,Duan+23,LeGall14,Stothers10,WilliamsSTOC12}. We say that a tensor $t'$ is a combinatorial restriction of~$t$ if $t'$ can be obtained from $t$ by zeroing out variables (i.e., setting some variables of $t$ to zero). 

The notion of degeneration can be used to define the notion of border rank of a tensor $t$, 
denoted $\underline{R}(t)$, as follows:
\[
\underline{R}(t)=\min\{{r\in\Nat\:|\: t\:\degen \:r \cdot\braket{1,1,1}\}}.
\]
The border rank can be used to give a formal definition of the exponent of matrix multiplication $\omega$ and more generally the exponent of rectangular matrix multiplication $\omega(\kappa)$ for any $\kappa\ge 0$:
\begin{align*}
\omega&=\inf\{\tau\in\Real\:|\:\underline{R}(\braket{n,n,n})=O(n^{\tau})\}\\
\omega(\kappa)&=\inf\{\tau\in\Real\:|\:\underline{R}\left(\braket{n,n,\ceil{n^k}}\right)=O(n^{\tau})\}.
\end{align*}

The border rank is submultiplicative: $\underline{R}(t\otimes t')\le \underline{R}(t)\times \underline{R}(t')$
for any two tensors $t$ and $t'$. This is the only property of the border rank we will use in this paper.

\paragraph{The asymptotic sum inequality.}
%%%%%%%%%%%%%%%%%%%%%
\hspace{-3mm} A powerful tool to derive upper bounds on $\omega(\kappa)$ is Sch{\"o}nhage asymptotic sum inequality \cite{Schonhage81}, which has been used in essentially all works on the exponent of square and rectangular matrix multiplication since its discovery in 1981. Here is the version we will use in this paper, which has also been used in prior works in rectangular matrix multiplication \cite{CoppersmithSICOMP82,Coppersmith97,Huang+98,Ke+08,LeGallFOCS12,LeGall+SODA18}.
\begin{proposition}\label{th:value}
Let $t$ be a tensor and $\kappa$ be a non-negative real number. If $t$ can be converted by a combinatorial restriction into a direct sum of $r$ terms, each isomorphic to $\braket{m,m,m^s}$ for some $s\ge \kappa$, then the following inequality holds: 
\[
r\cdot m^{\omega(\kappa)}\le \underline{R}(t).
\]
\end{proposition}

\paragraph{$\boldsymbol{\Ccc}$-tensors.}

%Finally, we will need the concept of decomposition of a tensor.
The concept of $\Ccc$-tensor was introduced by Strassen \cite{Strassen87}. Our treatment follows mainly \cite[Section 15.6]{Burgisser+97}. A $\Ccc$-tensor is a tensor that has an outer structure (called the support) isomorphic to a tensor. In this work, we will only consider the case where the outer structure is isomorphic to a matrix multiplication tensor.

Let $t\in U\otimes V\otimes W$ be a tensor. We say that $t$ is a $\Ccc$-tensor with support $\braket{e,h,\ell}$ if $U$, $V$ and $W$ can be decomposed as direct sums of subspaces 
$$
U=\bigoplus_{i=1}^e\bigoplus_{j=1}^h U_{i,j}
\hspace{3mm}
V=\bigoplus_{j=1}^h\bigoplus_{k=1}^\ell V_{j,k},
\hspace{3mm}
W=\bigoplus_{i=1}^e\bigoplus_{k=1}^\ell W_{i,k},
$$
and
$t$ can be written as 
$$t=\sum_{i=1}^e\sum_{j=1}^h\sum_{k=1}^\ell t_{ijk}$$ 
where each $t_{ijk}$ is a tensor in $U_{i,j}\otimes V_{j,k}\otimes W_{i,k}$.
The $t_{ijk}$'s are called the components of~$t$. 

As a simple example, consider the tensor 
\begin{equation}\label{eq:Ctensor}
t=t_{111}+t_{112} \:\:\:\textrm{ with }\:\:\:
t_{111}=\sum_{i=1}^p x_0\otimes y_i\otimes z_i \:\:\:\textrm{ and }\:\:\:
t_{112}=\sum_{i=1}^p x_0\otimes y'_i\otimes z'_i.
\end{equation}
By taking the decomposition $U=U_{1,1}$, $V=V_{1,1}\oplus V_{1,2}$ and $W=W_{1,1}\oplus W_{1,2}$, where $U_{1,1}$ has basis $\{x_0\}$, $V_{1,1}$ has basis $\{y_{i}\}$, $V_{1,2}$ has basis $\{y'_{i}\}$, $W_{1,1}$ has basis $\{z_{i}\}$, and $W_{1,2}$ has basis $\{z'_{i}\}$, we observe that this tensor is a $\Ccc$-tensor with support $\braket{1,1,2}$ in which each component is isomorphic to $\braket{1,1,p}$.

We first mention that the concept of $\Ccc$-tensor is preserved by the tensor product.
\begin{proposition}[Section 15.6 in \cite{Burgisser+97}]\label{prop:Ct}
Let $t$ be a $\Ccc$-tensor with support $\braket{e,h,\ell}$ in which each component is isomorphic to $\braket{m,n,p}$. Let $t'$ be a $\Ccc$-tensor with support $\braket{e',h',\ell'}$ in which each component is isomorphic to $\braket{m',n',p'}$. Then $t\otimes t'$ is a $\Ccc$-tensor with support $\braket{ee',hh',\ell\ell'}$ in which each component is isomorphic to $\braket{mm',nn',pp'}$.
\end{proposition}

The following proposition shows when both the support and the components are matrix multiplication tensors of the type $\braket{1,1,\cdot}$, the whole tensor is isomorphic to a matrix multiplication tensor of type $\braket{1,1,\cdot}$ as well.
\begin{proposition}\label{prop:Ct-mm}
Let $T$ be a $\Ccc$-tensor with support 
$\braket{1,1,\ell}$ in which each component is isomorphic to $\braket{1,1,p}$. Then $T\cong\braket{1,1,\ell p}$.
\end{proposition}
\begin{proof}
A $\Ccc$-tensor $T$ with support 
$\braket{1,1,\ell}$ in which each component is isomorphic to $\braket{1,1,p}$ can be written as
\[
T=\sum_{k=1}^\ell T_{11k} \:\:\:\textrm{ with }\:\:\:
T_{11k}=\sum_{i=1}^p x_0\otimes y_i^{k}\otimes z_i^{k} \:\:\textrm{ for all } \:\:k\in\set{1}{\ell},
\]
for $\ell p$ distinct $y$-variables $y_{i}^k$ and $\ell p$ distinct $z$-variables $z_{i}^k$ (Equation \eqref{eq:Ctensor} corresponds to the case $\ell=2$).
We thus have $T\cong\braket{1,1,\ell p}$.
\end{proof}

%======
\subsection{The Coppersmith-Winograd tensor}\label{sub:CW}
%======
For any positive integer $q$, the Coppersmith-Winograd tensor \cite{Coppersmith+90} is the tensor over $\field^{q+2}\otimes \field^{q+2}\otimes \field^{q+2}$ defined as
\begin{align*}
\tcw=\tcw^{[011]}+\tcw^{[101]}+\tcw^{[011]}+\tcw^{[002]}+\tcw^{[020]}+\tcw^{[200]},
\end{align*}
where 
\begin{flalign*}
\tcw^{[011]}=\sum_{i=1}^q x_0\otimes y_i\otimes z_i,&&
\tcw^{[101]}=\sum_{i=1}^q x_i\otimes y_0\otimes z_i,&&
\tcw^{[110]}=\sum_{i=1}^q x_i\otimes y_i\otimes z_0,&&\\
\tcw^{[002]}=x_0\otimes y_0\otimes z_{q+1},&&
\tcw^{[020]}=x_0\otimes y_{q+1}\otimes z_{0},&&
\tcw^{[200]}=x_{q+1}\otimes y_0\otimes z_0.&&
\end{flalign*}
Coppersmith and Winograd showed that $\underline{R}(\tcw)\le q+2$.

The square of this tensor, already studied in \cite{Coppersmith+90} (see also \cite{Alman+SODA21,Coppersmith+90,Davie+13,Duan+23,LeGall14,Stothers10,WilliamsSTOC12}), will be the starting block for our analysis of the fourth power. Define the set 
$$S_4=\{(ijk)\in\set{0}{4}^3 \mid i+j+k = 4\}.$$ 
By regrouping terms, we can write
$$\tcw^{\otimes 2} = \sum_{(ijk) \in S_4} \mathcal{T}_{ijk}$$
where
\begin{equation}\label{eq:decompT2}
\begin{split}
\mathcal{T}_{004}=&\tcw^{[002]}\otimes \tcw^{[002]},\\
\mathcal{T}_{013}=&\tcw^{[011]}\otimes \tcw^{[002]}+\tcw^{[002]}\otimes \tcw^{[011]},\\
\mathcal{T}_{022}=&\tcw^{[011]}\otimes \tcw^{[011]}+\tcw^{[002]}\otimes \tcw^{[020]}+\tcw^{[020]}\otimes \tcw^{[002]},\\
\mathcal{T}_{112}=&\tcw^{[011]}\otimes \tcw^{[101]}+\tcw^{[101]}\otimes \tcw^{[011]}+\tcw^{[002]}\otimes \tcw^{[110]}+\tcw^{[110]}\otimes \tcw^{[002]},
\end{split}
\end{equation}
and the other eleven terms are obtained by permuting the indices of the $x$-variables, the $y$-variables and $z$-variables in the above expressions. 

For any $(ijk) \in S_4 \setminus \{(112),(121),(211)\}$, the tensor $T_{ijk}$ represents a matrix product:
\begin{equation}\label{eq:term}
\begin{split}
\Tt_{004}\cong \Tt_{040}\cong \Tt_{400}\cong~&\langle 1,1,1\rangle,\\
\Tt_{013}\cong \Tt_{031}\cong~&\langle 1,1,2q\rangle,\\
\Tt_{103}\cong \Tt_{301}\cong~& \langle 2q,1,1\rangle,\\
\Tt_{130}\cong \Tt_{310}\cong~& \langle 1,2q,1\rangle,\\
\Tt_{022}\cong~&\langle 1,1,q^2+2\rangle,\\
\Tt_{202}\cong~&\langle q^2+2,1,1\rangle,\\
\Tt_{220}\cong~&\langle 1,q^2+2,1\rangle.
\end{split}
\end{equation}
%We thus get the bounds
%\begin{align*}
%\Val(\Tt_{004})=\Val(\Tt_{040})=\Val(\Tt_{400})&=0,\\
%\Val(\Tt_{013})=\Val(\Tt_{031})=\Val(\Tt_{103})=\Val(\Tt_{130})=\Val(\Tt_{301})=\Val(\Tt_{310})&
%\ge\frac{\rho}{3}\log(2q),\\
%\Val(\Tt_{022})=\Val(\Tt_{202})=\Val(\Tt_{220})&
%\ge\frac{\rho}{3} \log(q^2+2).
%\end{align*}
The other terms $\Tt_{112}$, $\Tt_{121}$ and $\Tt_{211}$ are not matrix multiplications, and require a more delicate analysis (an analysis tailored to our needs will be done in Section \ref{sub6}).

%=============
\section{Asymmetric Analysis of the Fourth Power}\label{sec:asym}
%=============
In this section we describe the asymmetric analysis of the fourth power from \cite{LeGall+SODA18}, which gives the upper bounds on $\omega(\kappa)$ reported in Table \ref{table_SODA18}. The analysis in Sections \ref{sub1}, \ref{sub2}, \ref{sub3}, \ref{sub4} and \ref{sub6} is the same as in \cite{LeGall+SODA18}, with only minor changes in the notation and presentation. In particular, we derive closed-form expressions for most quantities since this will be needed to analyze the combination loss in Sections \ref{sec:loss} and \ref{sec:rec}.
The analysis of the components in Section \ref{sub5}, on the other hand, is slightly different from the analysis in~\cite{LeGall+SODA18} since we explicitly state the expression of these components in terms of the parameters.\footnote{Ref.~\cite{LeGall+SODA18} did not introduce any parameters to analyze the components considered in Section \ref{sub5} since those components are isomorphic to matrix multiplication tensors and can then be analyzed in a straightforward way. In our work, however, we need to introduce parameters and reanalyze explicitly those components in terms of these parameters in order to calculate the combination loss in Sections \ref{sec:loss} and \ref{sec:rec}.}

%the optimal value of the component was found analytically.} 
%The presentation analysis of the subcomponents in \ref{sub6} is also slightly different than in \cite{LeGall+SODA18} since for the square case we can use tools (in particular, Proposition \ref{prop:Ct2}) that are not available when studying the rectangular case.

%===
\subsection{Decomposition into components}\label{sub1}
%===
Define the sets
\begin{align*}
S_8&=\left\{(ijk)\in\{0,\ldots,8\}^3\:\vert\: i+j+k=8\right\},\\
\bar S_8&=\left\{(ijk)\in\{1,\ldots,8\}^3\:\vert\: i+j+k=8 \right\},\\
S_8^\prec&=\left\{(ijk)\in\{0,\ldots,8\}^3\:\vert\: i+j+k=8\:\textrm{ and } i\le j\le k\right\},\\
\bar S^\prec_8&=\left\{(ijk)\in\{1,\ldots,8\}^3\:\vert\: i+j+k=8 \:\textrm{ and } i\le j\le k\right\}.
\end{align*}
%Note that $|S'_8|=25$. 
We decompose the fourth power of the Coppersmith-Winograd tensor as follows:
\[
\tcw^{\otimes 4}=\sum_{(ijk)\in S_8}T_{ijk}
\]
where
\begin{equation}\label{eq:decomposition}
T_{ijk}=\sum_{\substack{(abc),(a'b'c')\in S_4\\a+a'=i,\:b+b'=j,\:c+c'=k}}\Tt_{abc}\otimes \Tt_{a'b'c'}
\end{equation}
for each $(ijk)\in S_8$, where the $\Tt_{abc}$'s and $\Tt_{a'b'c'}$'s are defined in \eqref{eq:decompT2}. We call each $T_{ijk}$ a component, and call each $\Tt_{abc}\otimes \Tt_{a'b'c'}$ a subcomponent of $T_{ijk}$. 

%======
\subsection{Analysis of the first extraction}\label{sub2}
%======
For each $(ijk)\in S_8$ we introduce a variable $\alpha_{ijk}\in\Rat[0,1]$. We impose the following conditions
\begin{equation}\tag{C1}\label{eq:C1}
\sum_{(ijk)\in S_8} \alpha_{ijk}=1,
\end{equation}
\begin{equation}\tag{C2}\label{eq:C2}
\alpha_{ijk}=\alpha_{ikj} \:\:\textrm{ for all }\:\: (ijk)\in S_8,
\end{equation}
which reduce the number of free variables to 24. We additionally impose the following condition:
\begin{equation}\tag{C3}\label{eq:C3}
\begin{split}
	\alpha_{026}\alpha_{107}\alpha_{215} =& \alpha_{017}  \alpha_{125} \alpha_{206},\\
	\alpha_{026}\alpha_{107}\alpha_{611} = & \alpha_{017}  \alpha_{116} \alpha_{602},\\
	\alpha_{035} \alpha_{107} \alpha_{314}= & \alpha_{017}  \alpha_{134} \alpha_{305},\\
	\alpha_{044}\alpha_{107}\alpha_{413} = & \alpha_{017} \alpha_{134} \alpha_{404},\\
	\alpha_{035}\alpha_{107}\alpha_{503} = &\alpha_{017} \alpha_{125} \alpha_{512},\\
	\alpha_{035}\alpha_{107}\alpha_{116} \alpha_{224}=&\alpha_{017} \alpha_{125}  \alpha_{134} \alpha_{206},\\
	\alpha_{044}\alpha_{107} \alpha_{116}\alpha_{233}= &\alpha_{017} \alpha_{134} \alpha_{134} \alpha_{206},\\
	\alpha_{035} \alpha_{044} \alpha_{107} \alpha_{116}\alpha_{323}=  &\alpha_{017} \alpha_{026} \alpha_{134} \alpha_{134}\alpha_{305},\\
	\alpha_{044} \alpha_{035}\alpha_{107} \alpha_{116} \alpha_{323}= & \alpha_{017}  \alpha_{026} \alpha_{134} \alpha_{134} \alpha_{305},\\
	\alpha_{044} \alpha_{035} \alpha_{107} \alpha_{116}\alpha_{422}= &\alpha_{017} \alpha_{026} \alpha_{134}  \alpha_{125}\alpha_{404}.
\end{split}
\end{equation}

For each $\ell\in\set{0}{8}$, define 
\begin{align*}
A_\ell=\sum_{\substack{(ijk)\in S_8\\i=\ell}}\alpha_{ijk},\hspace{5mm}
B_\ell=\sum_{\substack{(ijk)\in S_8\\j=\ell}}\alpha_{ijk},\hspace{3mm}\textrm{and}\hspace{5mm}
C_\ell=\sum_{\substack{(ijk)\in S_8\\k=\ell}}\alpha_{ijk}.
\end{align*}
Under condition \eqref{eq:C1}, the set of parameters $\{\alpha_{ijk}\}$ can be considered as a probability distribution on the components.
We then define the probability distributions $\Aa,\Bb,\Cc\colon\set{0}{8}\to [0,1]$ as follows:
\begin{align*}
\Aa=&(A_0,A_1,A_2,A_3,A_4,A_5,A_6,A_7,A_8),\\
\Bb=&(B_0,B_1,B_2,B_3,B_4,B_5,B_6,B_7,B_8),\\
\Cc=&(C_0,C_1,C_2,C_3,C_4,C_5,C_6,C_7,C_8).
\end{align*}
These three probability distributions represent the marginal of (the indices of the) $x$-variables, $y$-variables and $z$-variables, respectively. As shown in \cite{LeGallFOCS12,LeGall+SODA18,WilliamsSTOC12}, Condition \eqref{eq:C3} ensures that the distribution $\{\alpha_{ijk}\}$ has maximum entropy among all distributions with marginals $\Aa,\Bb,\Cc$.

Note that Condition \eqref{eq:C2} implies that $B_\ell=C_\ell$ for all $\ell\in\set{0}{8}$, and thus $\Bb=\Cc$. We impose the condition 
\begin{equation}\tag{C4}\label{eq:C4}
\ent(\Aa)\ge\ent(\Bb).
\end{equation}

Section 3 of \cite{LeGall+SODA18} shows how to extract from the fourth power of the Coppersmith-Winograd tensor a direct sum of many tensors, each isomorphic to a tensor product of powers of components. More precisely, here is the main statement.\footnote{The proof in \cite[Section 3]{LeGall+SODA18} shows that $\tcw^{\otimes 4N}$ can be converted into a direct sum of 
\[
2^{\min\{\ent(\Aa),\ent(\Bb)\}(1-o(1))N}
\]
terms, each isomorphic to \eqref{eq:tensor}. We get our statement since $\min\{\ent(\Aa),\ent(\Bb)\}=\ent(\Bb)$, due to Condition \eqref{eq:C4}. Note that in \cite[Section 3]{LeGall+SODA18}, however, the condition adopted is $\ent(\Aa)\le\ent(\Bb)$, and thus $\min\{\ent(\Aa),\ent(\Bb)\}=\ent(\Aa)$ in the statement of Theorem 3.1 in \cite{LeGall+SODA18}.}
\begin{theorem}[Adapted from Theorem 3.1 in \cite{LeGall+SODA18}]\label{th:1}
For any set of parameters $\alpha_{ijk}\in\Rat[0,1]$ satisfying Conditions \eqref{eq:C1}, \eqref{eq:C2}, \eqref{eq:C3}, \eqref{eq:C4} and any large enough $N$, the tensor $\tcw^{\otimes 4N}$ can be converted by a combinatorial restriction into a direct sum of 
\[
2^{\ent(\Bb)(1-o(1))N}
\]
terms, each isomorphic to the tensor
\begin{equation}\label{eq:tensor}
\bigotimes_{(ijk)\in S_8}T_{ijk}^{\otimes \alpha_{ijk}N}.
\end{equation}
\end{theorem}

%==============
\subsection{Decomposition of components into subcomponents}\label{sub3}
%==============
Each component $T_{ijk}$ decomposes into subcomponents, as shown in Equation \eqref{eq:decomposition}. In this subsection, we introduce parameters to analyze this decomposition.

\paragraph{Definition of the parameters.}
For each component $T_{ijk}$, we introduce one parameter per subcomponent. The parameters for all components $T_{ijk}$ for $(ijk)\in S_8^\prec$ are shown in Table \ref{table:comp}. Note that in order to exploit the symmetries, we reduce the number of parameters as follows:
\begin{itemize}
\item
We systematically use the same parameter for two subcomponents $\Tt_{abc}\otimes \Tt_{a'b'c'}$ and $\Tt_{a'b'c'}\otimes \Tt_{abc}$ (when these subcomponents are distinct, i.e., when $(abc)\neq (a'b'c')$) of the same component. 
For instance, for the component
\[
T_{035}=\Tt_{004}\otimes \Tt_{031}+\Tt_{031}\otimes \Tt_{004}+\Tt_{013}\otimes \Tt_{022}+\Tt_{022}\otimes \Tt_{013},
\] 
instead of introducing four parameters, we introduce only two parameters $g_{0351},g_{0352}$ and assign $\frac{1}{2}g_{0351}$, $\frac{1}{2}g_{0351}$, $\frac{1}{2}g_{0352}$, $\frac{1}{2}g_{0352}$, respectively, to each of the four subcomponents.
\item
For the component $T_{224}$, we additionally require that the value of the parameter for the subcomponent $\Tt_{013}\otimes \Tt_{211}$ is equal to the value of the parameter for the subcomponent $\Tt_{103}\otimes \Tt_{121}$ (both are set to $\frac{1}{2}g_{2242}$). This condition ensures that the assignment of the parameters does not change when permuting the $x$-variables and $y$-variables of $T_{224}$ (which is needed to exploit the symmetry since we label this component with ``$224$'').
\item
For the component $T_{233}$, we additionally require that the value of the parameter for the subcomponent $\Tt_{013}\otimes \Tt_{220}$ is equal to the value of the parameter for the subcomponent $\Tt_{202}\otimes \Tt_{031}$ (both are set to $\frac{1}{2}g_{2331}$). This condition ensures that the assignment of the parameters does not change when permuting the $y$-variables and $z$-variables of $T_{233}$ (which is needed to exploit the symmetry since we label this component with ``$233$").
\end{itemize}
The parameters for all components $T_{ijk}$ with $(ijk)\in S_8\setminus S_8^\prec$ are obtained by permuting the indices. For instance, the four subcomponents of the component
\[
T_{503}=\Tt_{400}\otimes \Tt_{103}+\Tt_{103}\otimes \Tt_{400}+\Tt_{301}\otimes \Tt_{202}+\Tt_{202}\otimes \Tt_{301},
\]
are assigned parameters $\frac{1}{2}g_{5031}$, $\frac{1}{2}g_{5031}$, $\frac{1}{2}g_{5032}$ and $\frac{1}{2}g_{5032}$, respectively.

%%%%%%%%%%%%%%%%%
\begin{table}[!hbp]
\begin{center}
\begin{tabular}{|c||c|}
\hline
\multirow{4}{*}{$T_{008}$}&
 00\\
 &00\\
 &44\\
 \cline{2-2}
 &$g_{0081}$\\
\hline
\end{tabular}
\hspace{5mm}
\begin{tabular}{|c||cc|}
\hline
\multirow{4}{*}{$T_{017}$}&
 00&00\\
 &01&10\\
 &43&34\\
 \cline{2-3}
 &$\frac{1}{2}g_{0171}$&$\frac{1}{2}g_{0171}$\\
\hline
\end{tabular}
\hspace{5mm}
\begin{tabular}{|c||cc|c|}
\hline
\multirow{4}{*}{$T_{026}$}&
 00&00&00\\
 &02&20&11\\
 &42&24&33\\
 \cline{2-4}
 &$\frac{1}{2}g_{0261}$&$\frac{1}{2}g_{0261}$&$g_{0262}$\\
\hline
\end{tabular}

\vspace{4mm}
\begin{tabular}{|c||cc|cc|}
\hline
\multirow{4}{*}{$T_{035}$}&
 00&00&00&00\\
 &03&30&12&21\\
 &41&14&32&23\\
 \cline{2-5}
 &$\frac{1}{2}g_{0351}$&$\frac{1}{2}g_{0351}$&$\frac{1}{2}g_{0352}$&$\frac{1}{2}g_{0352}$\\
\hline
\end{tabular}
\hspace{5mm}
\begin{tabular}{|c||cc|cc|c|}
\hline
\multirow{4}{*}{$T_{044}$}&
 00&00&00&00&00\\
 &04&40&13&31&22\\
 &40&04&31&13&22\\
 \cline{2-6}
 &$\frac{1}{2}g_{0441}$&$\frac{1}{2}g_{0441}$&$\frac{1}{2}g_{0442}$&$\frac{1}{2}g_{0442}$&$g_{0443}$\\
\hline
\end{tabular}

\vspace{5mm}
\begin{tabular}{|c||cc|cc|}
\hline
\multirow{4}{*}{$T_{116}$}&
 01&10&01&10\\
 &01&10&10&01\\
 &42&24&33&33\\
 \cline{2-5}
 &$\frac{1}{2}g_{1161}$&$\frac{1}{2}g_{1161}$&$\frac{1}{2}g_{1162}$&$\frac{1}{2}g_{1162}$\\
\hline
\end{tabular}

\vspace{5mm}
\begin{tabular}{|c||cc|cc|cc|}
\hline
\multirow{4}{*}{$T_{125}$}&
 01&10&01&10&10&01\\
 &02&20&11&11&02&20\\
 &41&14&32&23&32&23\\
 \cline{2-7}
 &$\frac{1}{2}g_{1251}$&$\frac{1}{2}g_{1251}$&$\frac{1}{2}g_{1252}$&$\frac{1}{2}g_{1252}$&$\frac{1}{2}g_{1253}$&$\frac{1}{2}g_{1253}$\\
\hline
\end{tabular}

\vspace{5mm}
\begin{tabular}{|c||cc|cc|cc|cc|}
\hline
\multirow{4}{*}{$T_{134}$}&
  01&10&01&10&10&01&01&10\\
 &03&30&12&21&03&30&21&12\\
 &40&04&31&13&31&13&22&22\\
 \cline{2-9}
 &$\frac{1}{2}g_{1341}$&$\frac{1}{2}g_{1341}$&$\frac{1}{2}g_{1342}$&$\frac{1}{2}g_{1342}$&$\frac{1}{2}g_{1343}$&$\frac{1}{2}g_{1343}$&$\frac{1}{2}g_{1344}$&$\frac{1}{2}g_{1344}$\\
\hline
\end{tabular}

\vspace{5mm}
\begin{tabular}{|c||cc|cc|cc|cc|c|}
\hline
\multirow{4}{*}{$T_{224}$}&
 02&20&02&20&11&11&02&20&11\\
 &02&20&11&11&02&20&20&02&11\\
 &40&04&31&13&31&13&22&22&22\\
 \cline{2-10}
 &$\frac{1}{2}g_{2241}$&$\frac{1}{2}g_{2241}$&$\frac{1}{2}g_{2242}$&$\frac{1}{2}g_{2242}$&$\frac{1}{2}g_{2242}$&$\frac{1}{2}g_{2242}$&$\frac{1}{2}g_{2243}$&$\frac{1}{2}g_{2243}$&$g_{2244}$\\
\hline
\end{tabular}

\vspace{5mm}
\begin{tabular}{|c||cc|cc|cc|cc|cc|}
\hline
\multirow{4}{*}{$T_{233}$}&
  02&20&11&11& 02&20&20&02&11&11\\
 &12&21&03&30&21&12&03&30&12&21\\
 &30&03&30&03&21&12&21&12&21&12\\
 \cline{2-11}
 &$\frac{1}{2}g_{2331}$&$\frac{1}{2}g_{2331}$&$\frac{1}{2}g_{2332}$&$\frac{1}{2}g_{2332}$&$\frac{1}{2}g_{2333}$&$\frac{1}{2}g_{2333}$&$\frac{1}{2}g_{2331}$&$\frac{1}{2}g_{2331}$&$\frac{1}{2}g_{2334}$&$\frac{1}{2}g_{2334}$\\
\hline
\end{tabular}
\end{center}\vspace{-5mm}
\caption{Parameters for the components $T_{ijk}$ with $(ijk)\in S_8^\prec$. The first row shows the decomposition into subcomponents. The second row shows the parameters associated with each subcomponent.}\label{table:comp}
\end{table}
%%%%%%%%%%%%%%%%%

\paragraph{Conditions on the parameters.}
We impose the symmetry condition 
\begin{equation}\tag{D1}\label{eq:D1}
g_{ijk\ell}=g_{ikj\ell} \:\:\textrm{ for all }\:\: (ijk)\in S_8 \textrm{ and all } \ell.
\end{equation}
We thus have 64 free parameters $g_{ijk\ell}$. We require that for each component, the sum of all the parameters (for the component) is 1. This gives the following additional conditions on these parameters:
\begin{equation}\tag{D2}\label{eq:D2}
\begin{split}
g_{ijk1}&=1\:\:\textrm{ for all }\:(ijk)\in\{(008),(800),(017),(107),(701)\},\\
g_{ijk1}+g_{ijk2}&=1 \:\:\textrm{ for all }\:(ijk)\in\{(035),(305),(503),(026),\\
&\hspace{50mm}(206),(602),(116),(611)\},\\
g_{ijk1}+g_{ijk2}+g_{ijk3}&=1 \:\:\textrm{ for all }\:(ijk)\in\{(044),(404),(125),(215),(512)\},\\
g_{ijk1}+g_{ijk2}+g_{ijk3}+g_{ijk4}&=1 \:\:\textrm{ for all }\:(ijk)\in\{(134),(314),(413)\},\\
g_{ijk1}+2g_{ijk2}+g_{ijk3}+g_{ijk4}&=1 \:\:\textrm{ for all }\:(ijk)\in\{(224),(422)\},\\	
2g_{ijk1}+g_{ijk2}+g_{ijk3}+g_{ijk4}&=1 \:\:\textrm{ for all }\:(ijk)\in\{(233),(323)\}.
\end{split}
\end{equation}

\paragraph{Definition of the marginals.}
Under condition \eqref{eq:D2}, for each $(ijk)\in S_8$ the set of parameters $g_{ijk\ell}$'s can be considered as a probability distribution over the set of subcomponents of $T_{ijk}$. We denote this probability distribution by $g_{ijk}$. We define three probability distributions $\Aa_{ijk},\Bb_{ijk},\Cc_{ijk}\colon\set{1}{4}\to[0,1]$ corresponding to the marginals of the $x$-variables, the $y$-variables and the $z$-variables, respectively. Here are these distributions for all $(ijk)\in S_8^\prec$:

\begin{flalign*}
&
\left\{
\begin{array}{ll}
\Aa_{008}&=\left(1,0,0,0,0\right),\\
\Bb_{008}&=\left(1,0,0,0,0\right),\\
\Cc_{008}&=\left(0,0,0,0,1\right),
\end{array}
\right.
&&
\left\{
\begin{array}{ll}
\Aa_{017}&=\left(1,0,0,0,0\right),\\
\Bb_{017}&=\left(1/2,1/2,0,0,0\right),\\
\Cc_{017}&=\left(0,0,0,1/2,1/2\right),
\end{array}
\right.
&&\\
&
\left\{
\begin{array}{ll}
\Aa_{026}&=\left(1,0,0,0,0\right),\\
\Bb_{026}&=\left(\frac{g_{0261}}{2},g_{0262},\frac{g_{0261}}{2},0,0\right),\\
\Cc_{026}&=\left(0,0,\frac{g_{0261}}{2},g_{0262},\frac{g_{0261}}{2}\right),
\end{array}
\right.
&&
\left\{
\begin{array}{ll}
\Aa_{035}&=\left(1,0,0,0,0\right),\\
\Bb_{035}&=\left(\frac{g_{0351}}{2},\frac{g_{0352}}{2},\frac{g_{0352}}{2},\frac{g_{0351}}{2},0\right),\\
\Cc_{035}&=\left(0,\frac{g_{0351}}{2},\frac{g_{0352}}{2},\frac{g_{0352}}{2},\frac{g_{0351}}{2}\right),
\end{array}
\right.
\\
&
\left\{
\begin{array}{ll}
\Aa_{044}&=\left(1,0,0,0,0\right),\\
\Bb_{044}&=\left(\frac{g_{0441}}{2},\frac{g_{0442}}{2},g_{0443},\frac{g_{0442}}{2},\frac{g_{0441}}{2}\right),\\
\Cc_{044}&=\left(\frac{g_{0441}}{2},\frac{g_{0442}}{2},g_{0443},\frac{g_{0442}}{2},\frac{g_{0441}}{2}\right),
\end{array}
\right.
&&
\left\{
\begin{array}{ll}
\Aa_{116}&=\left(1/2,1/2,0,0,0\right),\\
\Bb_{116}&=\left(1/2,1/2,0,0,0\right),\\
\Cc_{116}&=\left(0,0,\frac{g_{1161}}{2},g_{1162},\frac{g_{1161}}{2}\right),
\end{array}
\right.
&&
\end{flalign*}\vspace{-10mm}

\begin{flalign*}
\left\{
\begin{array}{ll}
\Aa_{125}&=\left(1/2,1/2,0,0,0\right),\\
\Bb_{125}&=\left(\frac{g_{1251}+g_{1253}}{2},g_{1252},\frac{g_{1251}+g_{1253}}{2},0,0\right),\\
\Cc_{125}&=\left(0,\frac{g_{1251}}{2},\frac{g_{1252}+g_{1253}}{2},\frac{g_{1252}+g_{1253}}{2},\frac{g_{1251}}{2}\right),
\end{array}
\right.
&
&&
&&
\end{flalign*}\vspace{-10mm}

\begin{flalign*}
\left\{
\begin{array}{ll}
\Aa_{134}&=\left(1/2,1/2,0,0,0\right),\\
\Bb_{134}&=\left(\frac{g_{1341}+g_{1343}}{2},\frac{g_{1342}+g_{1344}}{2},\frac{g_{1342}+g_{1344}}{2},\frac{g_{1341}+g_{1343}}{2},0 \right),\\
\Cc_{134}&=\left(\frac{g_{1341}}{2},\frac{g_{1342}+g_{1343}}{2} ,g_{1344},\frac{g_{1342}+g_{1343}}{2},\frac{g_{1341}}{2}\right),
\end{array}
\right.
&&
\end{flalign*}\vspace{-10mm}

\begin{flalign*}
\left\{
\begin{array}{ll}
\Aa_{224}&=\left(\frac{g_{2241}+g_{2242}+g_{2243}}{2}, g_{2242}+g_{2244},\frac{g_{2241}+g_{2242}+g_{2243}}{2},0,0\right),\\
\Bb_{224}&=\left(\frac{g_{2241}+g_{2242}+g_{2243}}{2}, g_{2242}+g_{2244},\frac{g_{2241}+g_{2242}+g_{2243}}{2},0,0\right),\\
\Cc_{224}&=\left(\frac{g_{2241}}{2},g_{2242},g_{2243}+g_{2244},g_{2242},\frac{g_{2241}}{2}\right),
\end{array}
\right.
&&
\end{flalign*}\vspace{-10mm}

\begin{flalign*}
\left\{
\begin{array}{ll}
\Aa_{233}&=\left(g_{2331}+\frac{g_{2333}}{2}, g_{2332}+g_{2334}, g_{2331}+\frac{g_{2333}}{2},0,0\right),\\
\Bb_{233}&=\left(\frac{g_{2331}+g_{2332}}{2}, \frac{g_{2331}+g_{2353}+g_{2334}}{2},\frac{g_{2331}+g_{2353}+g_{2334}}{2},\frac{g_{2331}+g_{2332}}{2},0\right),\\
\Cc_{233}&=\left(\frac{g_{2331}+g_{2332}}{2},\frac{g_{2331}+g_{2333}+g_{2334}}{2} ,\frac{g_{2331}+g_{2333}+g_{2334}}{2},\frac{g_{2331}+g_{2332}}{2},0\right).
\end{array}
\right.
&&
\end{flalign*}

 By distributions $\Aa_{ijk}$, $\Bb_{ijk}$ and $\Cc_{ijk}$ for $(ijk)\in S_8\setminus S^\prec_8$ can be obtained by permuting the indices.

%==============
\subsection{Analysis of the second extraction: the components from $\boldsymbol{\bar S_8}$}\label{sub4}
%==============
In this subsection we explain how to analyze each component $T_{ijk}$ for $(ijk)\in \bar S_8$. 
%We divide the analysis in two: the components for $(ijk)\in \bar S_8$ and then the components for $(ijk)\in S_8\setminus \bar S_8$.

\paragraph{Definition of the functions $\boldsymbol{\varphi_{ijk}}$.}
For each $(ijk)\in \bar S_8$, we define a function $\varphi_{ijk}\colon S_4\to [0,1]$ using the variables $g_{ijk\ell}$. The following table gives the definition of $\varphi_{ijk}$ for each $(ijk)\in  S^\prec_8$ (the value $\varphi_{ijk}(abc)$ is written in the cell corresponding to the row labeled ``$\varphi_{ijk}$'' and the column labeled ``$abc$''; this value represents the ``weight'' of the tensor $\Tt_{abc}$ in the component $T_{ijk}$ with respect to the probability distribution $g_{ijk}$): 

\begin{center}
\begin{tabular}{|c|c|c|c|c|c|c|c|c|c|c|c|c|c|c|c|}
\cline{2-16}
\multicolumn{1}{c|}{}&\!\!004\!\!&\!040\!&\!400\!&\!\!013\!\!&\!\!103\!\!&\!\!031\!\!&\!\!130\!\!&\!301\!&\!310\!&\!\!022\!\!&\!\!202\!\!&\!\!220\!\!&\!\!112\!\!&\!\!121\!\!&\!\!211\\
\hline
$\varphi_{116}$&\!\!
$g_{1161}$\!\!&\!0\!&\!0\!&\!\!$g_{1162}$\!\!&\!\!$g_{1162}$\!\!&\!\!0\!\!&\!\!0\!\!&\!\!0\!\!&\!\!0\!\!&\!\!0\!\!&\!\!0\!\!&\!\!0\!\!&\!\!$g_{1161}$\!\!&\!\!0\!\!&\!\!0\\
\hline
$\varphi_{125}$&\!\!
$g_{1251}$\!\!&\!0\!&\!0\!&\!\!$g_{1252}$\!\!&\!\!$g_{1253}$\!\!&\!\!0\!\!&\!\!0\!\!&\!\!0\!\!&\!\!0\!\!&\!\!$g_{1253}$\!\!&\!\!0\!\!&\!\!0\!\!&\!\!$g_{1252}$\!\!&\!\!$g_{1251}$\!\!&\!\!0\\
\hline
$\varphi_{134}$&\!\!
$g_{1341}$\!\!&\!0\!&\!0\!&\!\!$g_{1342}$\!\!&\!\!$g_{1343}$\!\!&\!\!$g_{1343}$\!\!&\!\!$g_{1341}$\!\!&\!\!0\!\!&\!\!0\!\!&\!\!$g_{1344}$\!\!&\!\!0\!\!&\!\!0\!\!&\!\!$g_{1344}$\!\!&\!\!$g_{1342}$\!\!&\!\!0\\
\hline
$\varphi_{224}$&\!\!
$g_{2241}$\!\!&\!0\!&\!0\!&\!\!$g_{2242}$\!\!&\!\!$g_{2242}$\!\!&\!\!0\!\!&\!\!0\!\!&\!\!0\!\!&\!\!0\!\!&\!\!$g_{2243}$\!\!&\!\!$g_{2243}$\!\!&\!\!$g_{2241}$\!\!&\!\!2$g_{2244}$\!\!&\!\!$g_{2242}$\!\!&\!\!$g_{2242}$\\
\hline
$\varphi_{233}$&\!\!
0\!\!&\!0\!&\!0\!&\!\!$g_{2331}$\!\!&\!\!$g_{2332}$\!\!&\!\!$g_{2331}$\!\!&\!\!$g_{1332}$\!\!&\!\!0\!\!&\!\!0\!\!&\!\!$g_{2333}$\!\!&\!\!$g_{2331}$\!\!&\!\!$g_{2331}$\!\!&\!\!$g_{2334}$\!\!&\!\!$g_{2334}$\!\!&\!\!$g_{2333}$\\
\hline
\end{tabular}
\end{center}

By permuting the indices, we extend this definition $\varphi_{ijk}$ to all $(ijk)\in \bar S_8\setminus \bar S^\prec_8$. For instance, the function $\varphi_{413}$ is given in the following table.
\begin{center}
\begin{tabular}{|c|c|c|c|c|c|c|c|c|c|c|c|c|c|c|c|}
\cline{2-16}
\multicolumn{1}{c|}{}&\!\!400\!\!&\!004\!&\!040\!&\!\!301\!\!&\!\!310\!\!&\!\!103\!\!&\!\!013\!\!&\!130\!&\!031\!&\!\!202\!\!&\!\!220\!\!&\!\!022\!\!&\!\!211\!\!&\!\!112\!\!&\!\!121\\
\hline
$\varphi_{413}$&\!\!
$g_{4131}$\!\!&\!0\!&\!0\!&\!\!$g_{4132}$\!\!&\!\!$g_{4133}$\!\!&\!\!$g_{4133}$\!\!&\!\!$g_{4131}$\!\!&\!\!0\!\!&\!\!0\!\!&\!\!$g_{4134}$\!\!&\!\!0\!\!&\!\!0\!\!&\!\!$g_{4134}$\!\!&\!\!$g_{4132}$\!\!&\!\!0\\
\hline
\end{tabular}
\end{center}

\paragraph{Conditions.}
We impose the conditions
\begin{equation}\tag{D3}\label{eq:D3}
\sum_{(ijk)\in  \bar S_8}\alpha_{ijk}\ent(\Aa_{ijk})\ge 
\sum_{(ijk)\in \bar S_8}\alpha_{ijk}\ent(\Bb_{ijk}),
\end{equation}
\begin{equation}\tag{D4}\label{eq:D4}
\begin{split}
	g_{2333} \sqrt{g_{2332}}=&g_{2331}\sqrt{g_{2334}},\\
	g_{3233}\sqrt{g_{3232}}=&g_{3231}\sqrt{g_{3234}}.
\end{split}
\end{equation}
Condition \eqref{eq:D3} is similar to Condition \eqref{eq:C4} for the first extraction. Note however that this is a global condition on the sum of the entropies for $(ijk)\in \bar S_8$, rather than a condition on each $(ijk)$. 

Condition \eqref{eq:D4} is similar to Condition \eqref{eq:C3} for the first extraction: as shown in \cite{LeGall+SODA18}, this condition ensures that the distributions $g_{233}$ and $g_{323}$ have maximum entropy among all distributions with marginals $\Aa_{233},\Bb_{233},\Cc_{233}$ and $\Aa_{323},\Bb_{323},\Cc_{323}$, respectively. (For all $(ijk)\in \bar S_{8}\setminus\{(233),(232),$ $(322)\}$, there is no need to impose such a condition  since the distribution $g_{ijk}$ necessarily has maximum entropy.) 

\paragraph{Analysis of the value.}
Theorem 4.1 in \cite{LeGall+SODA18} gives the following statement.

\begin{theorem}[Adapted from Theorem 4.1 in \cite{LeGall+SODA18}]\label{th:2}
For any set of parameters $\alpha_{ijk},g_{ijk\ell}\in\Rat[0,1]$ satisfying conditions \eqref{eq:D1}, \eqref{eq:D2}, \eqref{eq:D3}  and \eqref{eq:D4}, and any large enough $N$, the tensor
\[
\bigotimes_{(ijk)\in \bar S_8}T_{ijk}^{\otimes \alpha_{ijk}N}
\]
can be converted by a combinatorial restriction into a direct sum of 
\[
\prod_{(ijk)\in  \bar S_8}2^{\alpha_{ijk}\ent(\Bb_{ijk})(1-o(1))N}
\]
terms, each isomorphic to the tensor 
\begin{equation}\label{eq:tensor2}
\bigotimes_{(ijk)\in \bar S_8} \bigotimes_{(abc)\in S_4}\Tt_{abc}^{\otimes \varphi_{ijk}(abc)\alpha_{ijk}N}.
\end{equation}
\end{theorem}

%==============
\subsection{Analysis of the second extraction: the components from $\boldsymbol{S_8\setminus \bar S_8}$}\label{sub5}
%==============
In this subsection we explain how to analyze each component $T_{ijk}$ for $(ijk)\in \bar S_8$.

\paragraph{Definitions of the relevant quantities.}
For each $(ijk)\in \{(008),(017),(026),(035),(044)\}$, observe that the distributions $\Bb_{ijk}$ and $\Cc_{ijk}$ defined in Section \ref{sub3} are identical up to a permutation of the set $\set{0}{4}$. In particular we have $\ent(\Bb_{ijk})=\ent(\Cc_{ijk})$. We denote this quantity by $R_{ijk}$. Here are the concrete expressions:
\begin{align*}
R_{008}&=\ent\left(\left(1,0,0,0,0\right)\right)=0,\\
R_{017}&=\ent\left(\left(1/2,1/2,0,0,0\right)\right)=1,\\
R_{026}&=\ent\left(\left(\frac{g_{0261}}{2},g_{0262},\frac{g_{0261}}{2},0,0\right)\right),\\
R_{035}&=\ent\left(\left(\frac{g_{0351}}{2},\frac{g_{0352}}{2},\frac{g_{0352}}{2},\frac{g_{0351}}{2},0\right)\right),\\
R_{044}&=\ent\left(\left(\frac{g_{0441}}{2},\frac{g_{0442}}{2},g_{0443},\frac{g_{0442}}{2},\frac{g_{0441}}{2}\right)\right).
\end{align*}%\vspace{-10mm}
We also define the following quantities:
%\begin{align*}
%W_{008}&=1\\
%W_{017}&=\Tt_{013}\\
%W_{026}&=\Tt_{022}^{\otimes g_{0261}}\otimes \Tt_{013}^{2g_{0262}},\\
%W_{035}&=\Tt_{031}^{\otimes g_{0351}}\otimes \Tt_{013}^{g_{0352}}\otimes \Tt_{022}^{g_{0352}},\\
%W_{044}&=\Tt_{013}^{\otimes g_{0442}}\otimes \Tt_{031}^{\otimes g_{0442}}\otimes \Tt_{022}^{\otimes 2g_{0443}}.
%\end{align*}
\begin{align*}
W_{008}&=1,\\
W_{017}&=2q,\\
W_{026}&=(2q)^{2g_{0262}}(q^2+2)^{g_{0261}},\\
W_{035}&=(2q)(q^2+2)^{g_{0352}},\\
W_{044}&=(2q)^{2g_{0442}} (q^2+2)^{2g_{0443}}.
\end{align*}

By permuting the indices, we extend these definitions of $R_{ijk}$ and $W_{ijk}$ to all $(ijk)\in S_8\setminus \bar S_8$.

\paragraph{Analysis of the value.}
The following theorem shows how to analyze all these components.\footnote{Theorem \ref{th:3} is related to Claim 7 in Section 5 of the full version of \cite{WilliamsSTOC12}. Claim 7 in the full version of \cite{WilliamsSTOC12}, however, does not introduce any parameters to analyze these tensors since these tensors are isomorphic to matrix multiplication tensors and can then be written in a closed form. In our work, however, we need to introduce parameters and reanalyze explicitly those tensors in terms of these parameters in order to calculate the combination loss in Sections \ref{sec:loss} and \ref{sec:rec}.}
\begin{theorem}\label{th:3}
For any $(ijk)\in S_8\setminus \bar S_8$ and any large enough $N$, the tensor 
$
T_{ijk}^{\otimes \alpha_{ijk}N}
$
can be converted by a combinatorial restriction into a tensor isomorphic to 
\[
\begin{cases}
\braket{1,1,2^{\alpha_{ijk}R_{ijk}(1-o(1))N}\cdot W_{ijk}^N}& \textrm{ if } i = 0,\\
\braket{2^{\alpha_{ijk}R_{ijk}(1-o(1))N}\cdot W_{ijk}^N,1,1}& \textrm{ if } j = 0,\\
\braket{1,2^{\alpha_{ijk}R_{ijk}(1-o(1))N}\cdot W_{ijk}^N,1}& \textrm{ if } k = 0.\\
\end{cases}
\]
\end{theorem}
\begin{proof}
We first consider the case $(ijk)\in \{(008),(017),(026),(035),(044)\}$, and the tensor $T_{ijk}^{\otimes M}$ for some integer $M$ large enough so that $\Bb_{ijk}(c)$ and $\Cc_{ijk}(c)$ are multiples of $1/M$ for all $c\in\set{1}{4}$ (this is possible since we are assuming that all $g_{ijk\ell}$'s are rational numbers). 

Each variable of $T_{ijk}^{\otimes M}$ can be indexed by a string in $(\set{0}{4}\times \set{0}{4})^M$ by concatenating the indices of the variables. More precisely, each $y$-variable is indexed by a string $[(r_1,j-r_1),(r_2,j-r_2),\ldots,(r_{M},j-r_{M})]$; each $z$-variable is indexed by a string $[(s_1,k-s_1),(s_2,k-s_2),$ $\ldots,(s_{M},k-s_{M})]$; each $x$-variable is indexed by the string $[(0,0),(0,0),\ldots,(0,0)]$. For instance, for 
\[
T_{017}^{\otimes M}=(\Tt_{004}\otimes \Tt_{013}+\Tt_{013}\otimes \Tt_{004})^{\otimes M}=(\Tt_{004}\otimes \Tt_{013})^{\otimes M}+\cdots+(\Tt_{013}\otimes \Tt_{004})^{\otimes M},
\]
we have $(r_u,j-r_u)\in\{(0,1),(1,0)\}$ and $(s_u,k-r_u)\in\{(3,4),(4,3)\}$ for each $u\in\set{1}{M}$. For the first term in the sum, the $y$-variables are indexed by $[(0,1),(0,1),\dots,(0,1)]$ and the $z$-variables by $[(4,3),(4,3),\dots,(4,3)]$. For the last term in the sum, the $y$-variables are indexed by $[(1,0),(1,0),\dots,(1,0)]$ and the $z$-variables by $[(3,4),(3,4),\dots,(3,4)]$.

We set to zero all $y$-variables except those such that the distribution of the $r_u's$ among the $M$ coordinates of the index matches the distribution $\Bb_{ijk}$. Similarly set to zero all $z$-variables except those such that the distribution of the $s_u's$ among the $M$ coordinates of the index matches the distribution $\Cc_{ijk}$. Since the label for the $x$-variable is unique, and $\Bb_{ijk}=\Cc_{ijk}$, we obtain a $\Ccc$-tensor with support $\braket{1,1,\Nn_{ijk}}$, where
\begin{align*}
\Nn_{ijk}&={M \choose \Bb_{ijk}(0)M,\Bb_{ijk}(1)M,\Bb_{ijk}(2)M,\Bb_{ijk}(3)M,\Bb_{ijk}(4)M}\\
&=\Theta\left(\frac{2^{\ent(\Bb_{ijk})M}}{M^2}\right),\\
&=\Theta\left(\frac{2^{R_{ijk}M}}{M^2}\right),
\end{align*}
where the approximation of the multinomial coefficient is done using Stirling's formula. Since each component of the $\Ccc$-tensor is isomorphic to $\braket{1,1,W_{ijk}^M}$, by Proposition \ref{prop:Ct-mm} the whole tensor is isomorphic to $\braket{1,1,\Nn_{ijk}W_{ijk}^M}$.

The results for the cases $j=0$ and $k=0$ follow by permuting the indices.
\end{proof}

%========================
\subsection{The third extraction}\label{sub6}
%========================
It remains to further analyze the tensor \eqref{eq:tensor2}. The terms $\Tt_{abc}$ for $(abc)\in S_4\setminus \{(112),(121),(211)\}$ correspond to matrix multiplications tensors, as shown in  \eqref{eq:term}. We will use the following result from \cite{LeGall+SODA18} to analyze $\Tt_{112}$, $\Tt_{121}$ and $\Tt_{211}$, where for conciseness we write
\[
\myb=\log\left((2b)^b(1-b)^{(1-b)}\right)
\:\:\textrm { and }\:\:
\lambda_{\tilde b}=\log\left( (2\tilde b)^{\tilde b}(1-\tilde b)^{(1-\tilde b)}\right).
\]

\begin{theorem}[Theorem 5.1 in \cite{LeGallFOCS12}]\label{th:4}
For any $a_{211},a_{112}\in\Rat[0,1]$ and any parameters $b,\tilde b\in[0,1]$ such that the 
inequality 
\[
a_{112}\lambda_{b}\le a_{211} \lambda_{\tilde b}
\]
holds, and any large enough~$N$, the tensor
\[
\Tt_{112}^{\otimes a_{112}N}\otimes \Tt_{121}^{\otimes a_{112}N}\otimes \Tt_{211}^{\otimes a_{211}N}
\]
can be converted by a combinatorial restriction into a direct sum of 
\[
2^{
\left(2a_{112}+(1-\lambda_{\tilde b})a_{211}\right)(1-o(1))N
}
\]
terms, each isomorphic to the tensor 
\[
\left\langle
q^{(a_{112}+a_{211}\tilde b)N},q^{(a_{112}+a_{211}\tilde b)N},q^{(2a_{112}b+a_{211}(1-\tilde b))N}
\right\rangle.
\]
\end{theorem}

%========================
\subsection{Upper bound on $\boldsymbol{\omega(\kappa)}$ and optimization}\label{sub7}
%========================
Let us define the quantities $\Gamma,\Delta_x$ and $\Delta_y$ as follows:
\begin{align*}
\Gamma=&
\sum_{(ijk)\in \bar S_8} \left(\ent(\Bb_{ijk})+(2 \varphi_{ijk}(112)+(1-\lambda_{\tilde b})\varphi_{ijk}(211)\right)\alpha_{ijk},\\
\Delta_{x}=&
\sum_{\substack{(ijk)\in S_8\setminus \bar S_8\\j=0}} 
\left(R_{ijk}+\log(W_{ijk})\right)\alpha_{ijk}
+
\!
\sum_{(ijk)\in \bar S_8} 
\Big((\varphi_{ijk}(103)+\varphi_{ijk}(301))\log(2q)+ \\[-5pt]
&\hspace{30mm}
+\varphi_{ijk}(202)\log(q^2+2)+
\left(\varphi_{ijk}(112)+
\varphi_{ijk}(211)\tilde b\right)\log q
\Big)\alpha_{ijk},\\
\Delta_{z}=&
\sum_{\substack{(ijk)\in S_8\setminus \bar S_8\\i=0}} 
\left(R_{ijk}+\log(W_{ijk})\right)\alpha_{ijk}
+
\!
\sum_{(ijk)\in \bar S_8} 
\Big((\varphi_{ijk}(013)+\varphi_{ijk}(031))\log(2q)+ \\[-5pt]
&\hspace{30mm}
+\varphi_{ijk}(022)\log(q^2+2)+
\left(2\varphi_{ijk}(112)b+
\varphi_{ijk}(211)(1-\tilde b)\right)\log q
\Big)\alpha_{ijk}.
\end{align*}

In order to guarantee that the condition in the statement of Theorem \ref{th:4} is satisfied, we also introduce another condition:
\begin{equation}\tag{E1}\label{eq:E1}
\sum_{(ijk)\in \bar S_8} \varphi_{ijk}(112)\alpha_{ijk}\lambda_{b}
\le 
\sum_{(ijk)\in \bar S_8} \varphi_{ijk}(211)\alpha_{ijk}\lambda_{\tilde b}.
\end{equation}
Finally, for any fixed $\kappa$, we introduce the condition 
\begin{equation}\tag{E2}\label{eq:E2}
\Delta_z\ge \kappa \Delta_x.
\end{equation}

%We impose the constraint
%\begin{equation}\tag{E1}\label{eq:E1}
%%\log{\left(2^b b^b (1-b)^{1-b}\right)}\cdot
%\sum_{(i,j,k)\in S^\ast_8} f_{ijk}(211)\alpha_{ijk}
%\ge 
%%\log{\left(2^b b^b (1-b)^{1-b}\right)}\cdot
%\sum_{(i,j,k)\in S^\ast_8} f_{ijk}(112)\alpha_{ijk}
%\end{equation}
%in order to satisfy the condition required in Theorem \ref{th:3}.

Combining Theorems \ref{th:1}, \ref{th:2}, \ref{th:3} and $\ref{th:4}$, we obtain the following theorem. 
\begin{theorem}\label{th:opt1}
Consider any $\kappa\ge 0$.
For any set of parameters $\alpha_{ijk},g_{ijk\ell}\in\Rat[0,1]$ and any $b,\tilde b\in[0,1]$ satisfying Conditions \eqref{eq:C1},  \eqref{eq:C2}, \eqref{eq:C3}, \eqref{eq:C4}, \eqref{eq:D1}, \eqref{eq:D2}, \eqref{eq:D3}, \eqref{eq:D4}, \eqref{eq:E1}, \eqref{eq:E2}, 
the upper bound
\[
\omega(\kappa)\le \frac{4\log{(q+2)}-\Gamma-\ent(\Bb)}{\Delta_x}.
\]
holds. 
\end{theorem}
\begin{proof}
Combining Theorems \ref{th:1}, \ref{th:2}, \ref{th:3} and $\ref{th:4}$, and using the identities \eqref{eq:term} to analyze $\Tt_{013}$, $\Tt_{031}$, $\Tt_{103}$, $\Tt_{130}$, $\Tt_{301}$, $\Tt_{310}$, $\Tt_{022}$, $\Tt_{202}$, $\Tt_{220}$, we get that for any large enough $N$, the tensor $\tcw^{\otimes 4N}$ can be converted by a combinatorial restriction into a direct sum of 
\[
2^{(\Gamma+\ent(\Bb))(N-o(N))}
\]
terms, each isomorphic to a tensor 
\[
\left\langle
2^{\Delta_x(N-o(N))},2^{\Delta_y(N-o(N))},2^{\Delta_z(N-o(N))}
\right\rangle,
\]
where $\Delta_y=\Delta_x$ due to the symmetry Conditions \eqref{eq:C2} and \eqref{eq:D1}.

Using the upper bound $\underline R(\tcw^{\otimes 4})\le \underline R(\tcw)^4\le (q+2)^{4}$ from Section \ref{sub:CW}, applying Proposition~\ref{th:value} (which can be done due to Condition \eqref{eq:E2}), we obtain the inequality.
\[
\left(\Gamma+\ent(\Bb)+\omega(\kappa)\Delta_x\right)(N-o(N))\le 4\log(q+2) N,
\]
Dividing each side of above inequality by $N$ and then taking the limit when $N$ goes to infinity we obtain
\[
\Gamma+\ent(\Bb)+\omega(\kappa)\Delta_x\le 4\log(q+2),
\] 
which gives the claimed inequality.
\end{proof}

We have implemented the optimization problem corresponding to Theorem \ref{th:opt1} in Maple. Concretely, for a given $\kappa\ge 0$, we implement a search over the parameters $\alpha_{ijk},g_{ijk\ell}, b, \tilde b$ satisfying all the conditions of Theorem \ref{th:opt1} in order to find the smallest possible value of $\rho$ so that the inequality
\begin{equation}\label{eq:rho}
\Gamma+\ent(\Bb)+\rho\Delta_x\ge 4\log(q+2)
\end{equation}
holds. This value of $\rho$ is necessarily an upper bound on $\omega(\kappa)$. We stress that we do not need to find the minimum value of $\rho$ satisfying \eqref{eq:rho}: \emph{any} $\rho$ satisfying \eqref{eq:rho} gives an upper bound on $\omega(\kappa)$. This makes the task of verifying our numerical results easy: we only need to check that \eqref{eq:rho} and all the constraints of Theorem \ref{th:opt1} are satisfied.

The file of the optimization program can be found at \cite{File}.  For instance, for $\kappa=2$ we obtain the upper bound 
\[
\omega(2)<3.251640,
\]
which exactly matches the bound obtained in \cite{LeGall+SODA18}. 

More precisely, the parameters giving this upper bound are shown Table \ref{table:opt1}. For these parameters, we have
\begin{align*}
\Gamma+\ent(\Bb)+3.251640\cdot \Delta_x&= 11.2294215...,\\
4\log(q+2)&=11.2294197...,
\end{align*}
which implies the bound $\omega<3.251640$. The file available at \cite{File} also includes a program to check that these parameters satisfy all the constraints of Theorem~\ref{th:opt1}, as well as these calculations.
We observe that for these parameters we have 
\begin{align*}
\ent(\Aa)&=2.14399...,\\
\ent(\Bb)&=2.14399...,\\
\sum_{(ijk)\in  \bar S_8}\alpha_{ijk}\ent(\Aa_{ijk})&=0.99086...,\\ 
\sum_{(ijk)\in \bar S_8}\alpha_{ijk}\ent(\Bb_{ijk})&=0.99086...,
\end{align*}
i.e., both Inequalities \eqref{eq:C4} and \eqref{eq:D3} are saturated. In Sections \ref{sec:loss} and \ref{sec:rec} we will show how to relax \eqref{eq:C4} and \eqref{eq:D3}, which will give a better bound on $\omega(\kappa)$.

%\paragraph{Remark.}
%We remark that \cite{LeGall14,WilliamsSTOC12} reported the slightly better upper bound $\omega<2.3729269$ for the analysis of the 
%\cite{LeGall14,LeGall+SODA18,WilliamsSTOC12}
% is $\omega<2.372927$. This better upper bound can be obtained from the above analysis by removing the condition \eqref{eq:C3} and dealing with it appropriately.
%===============
\section{Analysis of the Combination Loss for the Outer Structure}\label{sec:loss}
%===============
In this section we apply the technique from \cite{Duan+23} to relax Condition $\eqref{eq:D3}$ and obtain improved upper bounds on $\omega(\kappa)$.

In Section \ref{sub:loss-param} we first define the parameters introduced in \cite[Section 6]{Duan+23} to analyze the combination loss. In Section \ref{sub:analysis-loss} we relax Condition $\eqref{eq:D3}$ by exploiting the combination loss and give Theorem \ref{th:opt2}, which improves Theorem \ref{th:opt1}. 
%====
\subsection{Parameters for analyzing the combination loss}\label{sub:loss-param}
%====
We define below $\gamma,\alpha_{i\textrm{++}}$ and $\bar \alpha_{i\textrm{++}}$ and $\chi$ as in \cite[Section 6]{Duan+23}.\footnote{The term $\chi$ corresponds to $\log (\alpha_P)$ in \cite{Duan+23} (it is more convenient for us to use the logarithm of $\alpha_P$). Also note that Ref.~\cite{Duan+23} actually uses the notation $\alpha_{\textrm{++}k}$ and $\bar \alpha_{\textrm{++}k}$, instead of $\alpha_{i\textrm{++}}$ and $\bar \alpha_{i\textrm{++}}$. This is because in \cite{Duan+23} the symmetry is between the $x$-variables and the $y$-variables, while in our work we have symmetry between the $y$-variables and the $z$-variables.} To make the definitions easier to understand, we give closed-form formulas and examples for several of them. Through this subsection we assume that Conditions \eqref{eq:C2}, \eqref{eq:D1} and \eqref{eq:D2} holds.

\paragraph{The distribution $\boldsymbol{\gamma}$.}
Define the probability distribution $\gamma\colon\{0,1,\ldots,4\}\times \{0,1,\ldots,4\}\to [0,1]$ as 
\[
\gamma(c,d)=\sum_{\substack{(i,j,k)\in S_8\\i=c+d}} \alpha_{ijk} \cdot \Aa_{ijk}(c)
\]
for all $(c,d)\in \{0,1,\ldots,4\}\times \{0,1,\ldots,4\}$, where $\Aa_{ijk}$ is the distribution defined in Section~\ref{sub3}. This distribution corresponds to the marginal distribution of the $x$-variable (seen as an index in $\{0,1,\ldots,4\}\times \{0,1,\ldots,4\}$) obtained when assigning the distribution $\{\alpha_{ijk}\}$ on the components and then assigning the distribution $g_{ijk}$ on each subcomponents of $T_{ijk}$. For example, we have 
\[
\gamma(3,3)=\alpha_{611}\Aa_{611}(3)+\alpha_{602}\Aa_{602}(3)+\alpha_{620}\Aa_{620}(3),
\]
since an index $(3,3)$ for the $x$-variable can only come from the components $T_{611}$, $T_{602}$ and $T_{620}$: for $T_{611}$ it comes from the subcomponents $\Tt_{310}\otimes \Tt_{301}$ and $\Tt_{301}\otimes\Tt_{310}$ (which have total weight $\Aa_{611}(3))$, for $T_{602}$ it comes from the subcomponent $T_{301}\otimes T_{301}$ (which has weight $\Aa_{602}(3)$), and for $T_{620}$ it comes from the subcomponent $T_{310}\otimes T_{310}$ (which has weight $\Aa_{620}(3)$).

\paragraph{The terms $\boldsymbol{\alpha_{i\textrm{++}}}$ and $\boldsymbol{\bar \alpha_{i\textrm{++}}}$.}
For each $i\in\set{0}{8}$,
define the quantity 
\[
\alpha_{i\textrm{++}}=
\sum_{\substack{j,k>0\\(ijk)\in S_8}} \alpha_{ijk}.
\]
For any $i\in\set{0}{8}$ such that $\alpha_{i\textrm{++}}\neq 0$, additionally define the probability distribution $\bar \alpha_{i\textrm{++}}\colon\set{0}{4}\to[0,1]$ as
\[
\bar \alpha_{i\textrm{++}}=
\frac{1}{\alpha_{i\textrm{++}}}
\sum_{\substack{j,k>0\\(ijk)\in S_8}} \alpha_{ijk}\cdot \Aa_{ijk},
\]
which corresponds to the average marginal distribution of the $x$-variable, where the average is taken over all components $T_{ijk}$ (weighted by $\alpha_{ijk}$) with $j,k>0$. Concretely, we get $\alpha_{8\textrm{++}}=\alpha_{7\textrm{++}}=0$, $\alpha_{i\textrm{++}}\neq 0$ for $i\in\set{0}{6}$, and have
\begin{align*}
\bar \alpha_{0\textrm{++}}&=\frac{2\alpha_{017}\Aa_{017}+2\alpha_{026}\Aa_{026}+2\alpha_{035}\Aa_{035}+\alpha_{044}\Aa_{044}}{2\alpha_{017}+2\alpha_{026}+2\alpha_{035}+\alpha_{044}},\\
\bar \alpha_{1\textrm{++}}&=\frac{\alpha_{116}\Aa_{116}+\alpha_{125}\Aa_{125}+\alpha_{134}\Aa_{134}}{\alpha_{116}+\alpha_{125}+\alpha_{134}},\\
\bar \alpha_{2\textrm{++}}&=\frac{2\alpha_{215}\Aa_{215}+2\alpha_{224}\Aa_{224}+\alpha_{233}\Aa_{233}}{2\alpha_{215}+2\alpha_{224}+\alpha_{233}},\\
\bar \alpha_{3\textrm{++}}&=\frac{\alpha_{314}\Aa_{314}+\alpha_{323}\Aa_{323}}{\alpha_{314}+\alpha_{323}},\\
\bar \alpha_{4\textrm{++}}&=\frac{2\alpha_{413}\Aa_{413}+\alpha_{422}\Aa_{422}}{2\alpha_{413}+\alpha_{422}},\\
\bar \alpha_{5\textrm{++}}&=\Aa_{512},\\
\bar \alpha_{6\textrm{++}}&=\Aa_{611}.\\
\end{align*}

\paragraph{The term $\boldsymbol{\chi}$.}
Finally, define the quantity
\[
\chi=\ent(\Aa)-\ent(\gamma)+
\sum_{\substack{(ijk)\in S_8\\j=0\textrm{ or }k=0}} 
\alpha_{ijk}\cdot \ent(\Aa_{ijk})
+
\sum_{i=0}^8 \alpha_{i\textrm{++}}\cdot \ent(\bar \alpha_{i\textrm{++}}).
\]
It can be shown that  $\chi\le 0$: Lemma 6.7 in \cite{Duan+23} shows that the quantity $2^\chi$ represents the probability of a ``block'' to be ``compatible''. We thus have $2^\chi\in [0,1]$, and then $\chi\le 0$.

%\footnote{More precisely, Lemma 6.7 in \cite{Duan+23} shows that the quantity $\alpha_P=2^p$ represents the probability of a block to be compatible. We thus have $2^p\in [0,1]$, and then $p\le 0$.}

%===============
\subsection{Including the combination loss into the analysis of Section \ref{sec:asym}}\label{sub:analysis-loss}
%===============
Consider the new condition
\begin{equation}\tag{C4'}\label{eq:CC4}
\ent(\Aa)-\chi\ge\ent(\Bb),
\end{equation}
which is a relaxation of Condition \eqref{eq:C4}, since $\chi\le 0$.

The analysis of \cite[Section 6]{Duan+23} shows the following theorem.\footnote{See in particular Equations (24) and (25) at the end of Section 6.2 of \cite{Duan+23}. Here is the correspondence between the main terms in \cite{Duan+23} and the terms of our paper:
\begin{align*}
\alpha_{BX}&\hspace{5mm}\longleftrightarrow\hspace{5mm} 2^{\ent(B)},\\
\alpha_{BZ}&\hspace{5mm}\longleftrightarrow\hspace{5mm}2^{\ent(A)},\\
\alpha_{P}&\hspace{5mm}\longleftrightarrow\hspace{5mm}2^{\chi}.
\end{align*}
Note that $\min\left(\alpha_{BX},\frac{\alpha_{BZ}}{\alpha_P}\right)=2^{\ent(\Bb)}$ under Condition \eqref{eq:CC4}.
Also note that due to Condition \eqref{eq:C3} on the parameters, we have $\max_{\alpha'\in D_{\alpha}} \{\alpha'_N\}=\alpha_N$ in Equation (25) of~\cite{Duan+23}.}
\begin{mytheorem}{\ref{th:1}'}[Adapted from Section 6 in \cite{Duan+23}]
For any set of parameters $\alpha_{ijk}\in\Rat[0,1]$ satisfying Conditions \eqref{eq:C1}, \eqref{eq:C2}, \eqref{eq:C3}, \eqref{eq:CC4} and any large enough $N$, the tensor $\tcw^{\otimes 4N}$ can be converted by a combinatorial restriction into a direct sum of 
\[
2^{\ent(\Bb)(1-o(1))N}
\]
terms, each isomorphic to the tensor
\[
\bigotimes_{(ijk)\in S_8}T_{ijk}^{\otimes \alpha_{ijk}N}.
\]
\end{mytheorem}
\addtocounter{theorem}{-1}
Note that the only difference between Theorem \ref{th:1}' and Theorem \ref{th:1} is that Condition  \eqref{eq:C4} is relaxed to Condition \eqref{eq:CC4}.

Replacing Theorem \ref{th:1} by Theorem \ref{th:1}' in the analysis of Section \ref{sub7}, we immediately obtain the following theorem.
\begin{theorem}\label{th:opt2}
Consider any $\kappa\ge 0$. For any set of parameters $\alpha_{ijk},g_{ijk\ell}\in\Rat[0,1]$ and any $b,\tilde b\in[0,1]$ satisfying Conditions \eqref{eq:C1},  \eqref{eq:C2}, \eqref{eq:C3}, \eqref{eq:CC4}, \eqref{eq:D1}, \eqref{eq:D2}, \eqref{eq:D3}, \eqref{eq:D4}, \eqref{eq:E1}, \eqref{eq:E2} the upper bound
\[
\omega(\kappa)\le \frac{4\log{(q+2)}-\Gamma-\ent(\Bb)}{\Delta_x}
\]
holds.
\end{theorem}

The only difference between Theorem \ref{th:opt1} and Theorem \ref{th:opt2} is again that Condition  \eqref{eq:C4} is relaxed to Condition \eqref{eq:CC4}.

We have implemented the optimization problem corresponding to Theorem \ref{th:opt2} in Maple. The search is similar that the search done in Section \ref{sub7}, but imposes Condition \eqref{eq:CC4} instead of \eqref{eq:C4} on the parameters. The file of the optimization program can be found at \cite{File}.  For $\kappa=2$, we obtain the upper bound 
\[
\omega(2)<3.251502,
\]
which improves the upper bound $\omega<3.251640$ from Section \ref{sec:asym} (i.e., the bound from \cite{LeGall+SODA18}). 

The parameters giving this upper bound are shown in Table \ref{table:opt2}. For these parameters, we have
\begin{align*}
\Gamma+\ent(\Bb)+3.251502\cdot \Delta_x&= 11.2294199...,\\
4\log(q+2)&=11.2294197...,
\end{align*}
which implies the bound $\omega(2)<3.251502$. 
The file available at \cite{File} also includes a program to check that these parameters satisfy all the constraints of Theorem~\ref{th:opt2}, as well as these calculations. 
We observe that for these parameters we have 
\begin{align*}
\ent(\Aa)&=2.14147...,\\
\ent(\Bb)&=2.14469...,\\
\sum_{(ijk)\in  \bar S_8}\alpha_{ijk}\ent(\Aa_{ijk})&=0.98800...,\\ 
\sum_{(ijk)\in \bar S_8}\alpha_{ijk}\ent(\Bb_{ijk})&=0.98800...,\\
-\chi&=0.00322..,
\end{align*}
and thus Inequalities \eqref{eq:CC4} and \eqref{eq:D3} are saturated. 
In Section \ref{sec:rec} we will show how to relax \eqref{eq:D3}, which will give a better bound on $\omega(\kappa)$.

\section{Analysis of the Combination Loss for the Components}\label{sec:rec}
%================
In addition to the combination loss considered in Section \ref{sec:loss}, we consider in this section the combination loss in the analysis of the components $T_{ijk}$ as well. This enables us to relax Condition $\eqref{eq:D3}$ and obtain better upper bounds on $\omega(\kappa)$. 

The approach to analyze recursively the combination loss for components of powers of the Coppersmith-Winograd tensor is developed in its full generality in \cite[Section 7]{Duan+23}. In Section~\ref{sub:rec-param} we first define the parameters needed for the analysis. Since we are working only with the fourth power of the Coppersmith-Winograd tensor, we can give closed-form formulas for most of them. In Section \ref{sub:rec-analysis} we apply the analysis of the combination loss from \cite[Section 7]{Duan+23}, combine it with the approach we introduced in Section~\ref{sec:asym} in order to relax Condition~$\eqref{eq:D3}$, and give our new bounds on $\omega(\kappa)$.
%============
\subsection{Parameters for the analysis}\label{sub:rec-param}
%============

For any $(ijk)\in \bar S_8$, we introduce several parameters defined in \cite[Section 7]{Duan+23} to analyze the combination loss of $T_{ijk}$.

%We first define five distributions $\Dd_{1}, \Dd_{2}, \Dd_{3}, \Dd_{4}, \Dd_{5}\colon\left\{0,1,2\right\}^2\to[0,1]$ as follows:
%\begin{align*}
%\Dd^1&=1,\\ %\textrm{ and } \Dd_{400}(a,b)=0 \textrm{ for all }(a,b)\neq(2,2),\\
%\Dd^2(0)&=1,\\ %\textrm{ and } \Dd_{004}(a,b)=0 \textrm{ for all }(a,b)\neq(0,0),\\
%\Dd_{103}(0,1)&=\Dd_{103}(1,0)=1/2,\\ %  \textrm{ and } \Dd_{103}(a,b)=0 \textrm{ for all }(a,b)\notin\{(0,1),(1,0)\},\\
%\Dd_{211}(0,2)&=\Dd_{211}(2,0)=(1-b)/2, \:\:\Dd_{211}(1,1)=b,\\
%\Dd_{202}(0,2)&=\Dd_{202}(2,0)=1/(2+q^2), \:\:\Dd_{202}(1,1)=q^2/(2+q^2).
%\end{align*}
\paragraph{Distributions $\boldsymbol{\Dd_{rst}}$.}
For each $(rst)\in S_4$, define the probability distribution $\Dd_{rst}\colon\left\{0,1,2\right\}^2\to[0,1]$ as follows: 
\begin{equation}\label{eq:Dd}
\begin{split}
\Dd_{400}&=(0,0,1),\\
\Dd_{211}&=\left(\frac{1-\tilde b}{2}, \tilde b,\frac{1-\tilde b}{2}\right),\\
\Dd_{310}=\Dd_{301}&=\left(0,\frac{1}{2},\frac{1}{2}\right),\\
\Dd_{220}=\Dd_{202}&=\left(\frac{1}{2+q^2},\frac{q^2}{2+q^2},\frac{1}{2+q^2}\right),\\
\Dd_{004}=\Dd_{040}=\Dd_{013}=\Dd_{031}=\Dd_{022}&=(1,0,0),\\
\Dd_{130}=\Dd_{103}=\Dd_{112}=\Dd_{121}&=\left(\frac{1}{2},\frac{1}{2},0\right),
\end{split}
\end{equation}
where $\tilde b\in[0,1]$ is the parameter used in Section \ref{sub6}.\footnote{The distribution $\Dd_{rst}$ actually corresponds to the marginal distribution of (the first half of) the $x$-variables when analyzing the tensor $\Tt_{rst}$. For instance, for  
\[
\mathcal{T}_{301}=\tcw^{[101]}\otimes \tcw^{[201]}+\tcw^{[201]}\otimes \tcw^{[101]},
\]
each of the two terms are assigned weight $1/2$ and thus the distribution is $(0,1/2,1,2)$.
For $\Tt_{211}$, $\Tt_{220}$ and $\Tt_{202}$, there is a degree of freedom in the choice of the distribution. The distributions $\Dd_{211}$, $\Dd_{220}$ and $\Dd_{202}$ given in Equation~\eqref{eq:Dd} seem to give the best upper bound on $\omega(\kappa)$.} 

\paragraph{The distributions $\boldsymbol{\gamma_{ijk}}$.} For each $(ijk)\in \bar S_8$, the distribution $\gamma_{ijk}\colon\{0,1,2\}^4\to [0,1]$ is defined as 
\[
\gamma_{ijk}(a,b,c,d)=\sum_{\substack{(rst)\in S_4\\r=a+b\\r'=c+d}} \frac{\varphi_{ijk}(rst)}{2} \cdot \Dd_{rst}(a)\cdot \Dd_{r's't'}(c)
\]
for each $(a,b,c,d)\in \{0,1,2\}^4$, where in the sum we use the notation $r'=i-r$, $s'=j-s$ and $t'=k-t$. 

The distribution $\gamma_{ijk}$ corresponds to the marginal distribution of the $x$-variables (seen as elements of $\{0,1,2\}^4$) when applying the distribution $g_{ijk}$ on the subcomponents of $T_{ijk}$ and then decomposing each term $\Tt_{rst}$ using the distribution $\Dd_{rst}$. For instance, for 
\[
T_{116}=\Tt_{004}\otimes \Tt_{112}+\Tt_{112}\otimes \Tt_{004}+\Tt_{013}\otimes \Tt_{103}+\Tt_{103}\otimes \Tt_{013},
\]
we have $\gamma_{116}(0,1,0,0)=\gamma_{116}(1,0,0,0)=\gamma_{116}(0,0,1,0)=\gamma_{116}(0,0,0,1)=\frac{g_{1161}}{2}\frac{1}{2}+\frac{g_{1162}}{2}\frac{1}{2}=\frac{1}{4}$ and $\gamma_{116}(a,b,c,d)=0$ for all $(a,b,c,d)\notin\{(1,0,0,0),(0,1,0,0),(0,0,1,0),(0,0,0,1)\}$, since the $x$-variables of $\Tt_{004}$ and $\Tt_{013}$ are $00$ with probability 1 and the $x$-variables of $\Tt_{112}$ and $\Tt_{103}$ are $01$ or $10$ (each with probability $1/2$).    

\paragraph{The terms $\boldsymbol{\beta_{ijk}^{r\textrm{++}}}$ and $\boldsymbol{\bar \beta_{ijk}^{r\textrm{++}}}$.}
For each $(ijk)\in \bar S_8$ and each $r\in\set{0}{4}$,
define the quantity 
\[
\beta_{ijk}^{r\textrm{++}}=
\sum_{\substack{s,t>0\\(rst)\in S_4}} \frac{\varphi_{ijk}(rst)}{2}.
\]
If $\beta_{ijk}^{r\textrm{++}}>0$ we also define 
the distribution $\bar \beta^{r\textrm{++}}_{ijk}\colon\set{0}{2}\to[0,1]$ as
\[
\bar \beta_{ijk}^{r\textrm{++}}=
\frac{1}{\beta^{r\textrm{++}}_{ijk}}
\sum_{\substack{s,t>0\\(rst)\in S_4}} \frac{\varphi_{ijk}(rst)}{2}\cdot \Dd_{rst},
\]
which corresponds to the average marginal distribution of the (first half of the) $x$-variables, where the average is taken over the subcomponents $\Tt_{rst}\otimes \Tt_{r's't'}$ of $T_{ijk}$ with $s,t>0$.
Table~\ref{table:beta} gives the values of $\beta_{ijk}^{r\textrm{++}}$ and $\bar \beta^{r\textrm{++}}_{ijk}$ for all $(ijk)\in \bar S_8$ with $j\le k$ and all $r\in\set{0}{4}$ (for the case $k>j$, note that under Condition \eqref{eq:D1} we have $\beta_{ijk}^{r\textrm{++}}=\beta_{ikj}^{r\textrm{++}}$ and $\bar \beta_{ijk}^{r\textrm{++}}=\bar \beta_{ikj}^{r\textrm{++}}$).

\begin{table}[!tbp]
\begin{center}
\begin{tabular}{|c|c|c|c|c|c|c|c|c|c|c|}
\cline{2-11}
%\multicolumn{1}{c|}{}&
%\multicolumn{2}{c|}{$r=0$}&
%\multicolumn{2}{c|}{$r=1$}&
%\multicolumn{2}{c|}{$r=2$}\\
\multicolumn{1}{c|}{}&
$\beta_{ijk}^{0\textrm{++}}$&$\bar\beta_{ijk}^{0\textrm{++}}$
&
$\beta_{ijk}^{1\textrm{++}}$&$\bar\beta_{ijk}^{1\textrm{++}}$
&
$\beta_{ijk}^{2\textrm{++}}$&$\bar\beta_{ijk}^{2\textrm{++}}$
&
$\beta_{ijk}^{3\textrm{++}}$&$\bar\beta_{ijk}^{3\textrm{++}}$
&
$\beta_{ijk}^{4\textrm{++}}$&$\bar\beta_{ijk}^{4\textrm{++}}$
\bigstrut\\
%\cline{2-7}
\hline
116&$\frac{g_{1161}}{2}$&$\Dd_0$&$\frac{g_{1161}}{2}$&$\Dd_1$&0&-&0&-&0&-\bigstrut\\
\hline
125&$\frac{g_{1252}+g_{1253}}{2}$&$\Dd_0$&$\frac{g_{1251}+g_{1252}}{2}$&$\Dd_1$&0&-&0&-&0&-\bigstrut\\
\hline
134&$\frac{g_{1342}+g_{1343}+g_{1344}}{2}$&$\Dd_0$&$\frac{g_{1342} + g_{1344}}{2}$&$\Dd_1$&0&-&0&-&0&-\bigstrut\\
\hline
215&$\frac{g_{2153}}{2}$&$\Dd_0$&$\frac{g_{2152}}{2}$&$\Dd_1$&$\frac{g_{2151}}{2}$&$\Dd_2$&0&-&0&-\bigstrut\\
\hline
224&$\frac{g_{2242} + g_{2243}}{2}$&$\Dd_0$&$\frac{g_{2242}}{2} + g_{2244}$&$\Dd_1$&$\frac{g_{2242}}{2}$&$\Dd_2$&0&-&0&-\bigstrut\\
\hline
233&$g_{2331} + \frac{g_{2333}}{2}$&$\Dd_0$&$g_{2334}$&$\Dd_1$&$\frac{g_{2333}}{2}$&$\Dd_2$&0&-&0&-\bigstrut\\
\hline
314&$\frac{g_{3143}}{2}$&$\Dd_0$&$\frac{g_{3144}}{2}$&$\Dd_1$&$\frac{g_{3142}}{2}$&$\Dd_2$&0&-&0&-\bigstrut\\
\hline
323&$\frac{g_{3232} + g_{3231}}{2}$&$\Dd_0$&$\frac{g_{3233} + g_{3234}}{2}$&$\Dd_1$&$\frac{g_{3234}}{2}$&$\Dd_2$&0&-&0&-\bigstrut\\
\hline
413&$\frac{g_{4131}}{2}$&$\Dd_0$&$\frac{g_{4132}}{2}$&$\Dd_1$&$\frac{g_{4134}}{2}$&$\Dd_2$&0&-&0&-\bigstrut\\
\hline
422&$\frac{g_{4221}}{2}$&$\Dd_0$&$g_{4222}$&$\Dd_1$&$g_{4224}$&$\Dd_2$&0&-&0&-\bigstrut\\
\hline
512&0&-&$\frac{g_{5121}}{2}$&$\Dd_1$&$\frac{g_{5122}}{2}$&$\Dd_2$&0&-&0&-\bigstrut\\
\hline
611&0&-&0&-&$\frac{g_{6111}}{2}$&$\Dd_2$&0&-&0&-\bigstrut\\
\hline
\end{tabular}\vspace{-5mm}
\end{center}
\caption{The terms $\beta_{ijk}^{r\textrm{++}}$ and the distributions $\bar \beta_{ijk}^{r\textrm{++}}$ in Section \ref{sub:rec-param} for $j\le k$. The symbol ``-" means that the distribution is not defined (since $\beta_{ijk}^{r\textrm{++}}=0$). In this table we use $\Dd_0$, $\Dd_1$ and $\Dd_2$ to denote the following three probability distributions over $\{0,1,2\}$: $\Dd_0= (1,0,0)$, $\Dd_1= \left(\frac{1}{2},\frac{1}{2},0\right)$, $\Dd_2= \left(\frac{1-\tilde b}{2},\tilde b,\frac{1-\tilde b}{2}\right)$.
}
\label{table:beta}
\end{table}

\paragraph{The term $\boldsymbol{\chi_{ijk}}$.}
Define the quantity
\[
\chi_{ijk}=\ent(\Aa_{ijk})-\ent(\gamma_{ijk})+
\sum_{\substack{(rst)\in S_4\\s=0\textrm{ or }t=0}} \varphi_{ijk}(rst)\cdot \ent(\Dd_{rst})
+
\sum_{r=0}^4 2\beta_{ijk}^{r\textrm{++}}\cdot \ent\left(\bar \beta_{ijk}^{r\textrm{++}}\right).
\]

This quantity is similar to $\chi$ in Section \ref{sec:loss}, and represents the combination loss in the analysis of $T_{ijk}$.
It can be shown (see Lemma 7.12 in \cite{Duan+23}) that  $2^{\chi_{ijk}}\in [0,1]$, and thus $\chi_{ijk}\le 0$.

%====
\subsection{Including the combination loss of components into the analysis of Section \ref{sec:asym}}\label{sub:rec-analysis}
%====
Consider the new condition
\begin{equation}\tag{D3'}\label{eq:DD3}
\sum_{(ijk)\in  \bar S_8}\alpha_{ijk}\left(\ent(\Aa_{ijk})-\chi_{ijk}\right)\ge 
\sum_{(ijk)\in \bar S_8}\alpha_{ijk}\ent(\Bb_{ijk}),
\end{equation}
which is a relaxation of Condition \eqref{eq:D3}, since $\chi_{ijk}\le0$ for all $(ijk)\in \bar S_8$.

The analysis of \cite[Section 7]{Duan+23} shows the following theorem.\footnote{See in particular Equations (33) and (34) at the end of Section 7.2 of \cite{Duan+23}. Here is the correspondence between the main terms in \cite{Duan+23} and the terms of our paper:
\begin{align*}
\alpha_{BX}&\hspace{5mm}\longleftrightarrow\hspace{5mm} \prod_{(ijk)\in  \bar S_8} 2^{\alpha_{ijk}\ent(B_{ijk})},\\
\alpha_{BZ}&\hspace{5mm}\longleftrightarrow\hspace{5mm}\prod_{(ijk)\in  \bar S_8} 2^{\alpha_{ijk}\ent(A_{ijk})},\\
\alpha_{P}&\hspace{5mm}\longleftrightarrow\hspace{5mm} \prod_{(ijk)\in  \bar S_8}2^{\alpha_{ijk}\chi_{ijk}}.
\end{align*}
Note that 
\[
\min\left(\alpha_{BX},\frac{\alpha_{BZ}}{\alpha_P}\right)=\prod_{(ijk)\in  \bar S_8} 2^{\alpha_{ijk}\ent(B_{ijk})}
\] 
under Condition \eqref{eq:DD3}.
Also note that due to Condition \eqref{eq:D4} on the parameters, we have $\max_{\alpha'\in D_{\alpha}} \{\alpha'_N\}=\alpha_N$ in Equation (34) of~\cite{Duan+23}.}

\begin{mytheorem}{\ref{th:2}'}[Adapted from Section 7 in \cite{Duan+23}]
For any set of parameters $\alpha_{ijk},g_{ijk\ell}\in\Rat[0,1]$ satisfying conditions \eqref{eq:D1}, \eqref{eq:D2}, \eqref{eq:DD3}  and \eqref{eq:D4}, and any large enough $N$, the tensor
\[
\bigotimes_{(ijk)\in \bar S_8}T_{ijk}^{\otimes \alpha_{ijk}N}
\]
can be converted by a combinatorial restriction into a direct sum of 
\[
\prod_{(ijk)\in  \bar S_8}2^{\alpha_{ijk}\ent(\Bb_{ijk})(1-o(1))N}
\]
terms, each isomorphic to the tensor 
\[
\bigotimes_{(ijk)\in \bar S_8} \bigotimes_{(abc)\in S_4}\Tt_{abc}^{\otimes \varphi_{ijk}(abc)\alpha_{ijk}N}.
\]
\end{mytheorem}
\addtocounter{theorem}{-1}

Note that the only difference between Theorem \ref{th:2}' and Theorem \ref{th:2} is that Condition  \eqref{eq:D3} is relaxed to Condition \eqref{eq:DD3}.

Replacing Theorems \ref{th:1} and \ref{th:2} by Theorems \ref{th:1}' and \ref{th:2}', respectively, in the analysis of Section~\ref{sub7}, we immediately obtain the following theorem.

\begin{theorem}\label{th:opt3}
Consider any $\kappa\ge 0$. For any set of parameters $\alpha_{ijk},g_{ijk\ell}\in\Rat[0,1]$ and any $b,\tilde b\in[0,1]$ satisfying Conditions \eqref{eq:C1},  \eqref{eq:C2}, \eqref{eq:C3}, \eqref{eq:CC4}, \eqref{eq:D1}, \eqref{eq:D2}, \eqref{eq:DD3}, \eqref{eq:D4}, \eqref{eq:E1}, \eqref{eq:E2}  the upper bound
\[
\omega(\kappa)\le \frac{4\log{(q+2)}-\Gamma-\ent(\Bb)}{\Delta_x}
\]
holds.
\end{theorem}

The only difference between Theorem \ref{th:opt2} and Theorem \ref{th:opt3} is that Condition \eqref{eq:D3} is relaxed to Condition \eqref{eq:DD3}.

We have implemented the optimization problem corresponding to Theorem \ref{th:opt3} in Maple. The search is similar that the search done in Section \ref{sub:analysis-loss}, but imposes Condition \eqref{eq:DD3} instead of \eqref{eq:D3} on the parameters. The file of the optimization program can be found at \cite{File}.  For $\kappa=2$, we obtain the upper bound 
\[
\omega<3.250563,
\]
which further improves the upper bound  from Section \ref{sec:loss}. More generally, we obtain the upper bounds shown in Table \ref{table_ours} in Section \ref{sec:intro}.

The parameters giving this upper bound for $\omega(2)$ are shown in Table \ref{table:opt3} (the parameters for all the others values of $\kappa$ can be found at \cite{File}). For these parameters, we have
\begin{align*}
\Gamma+\ent(\Aa)+3.250563\cdot\Delta_x&= 11.229435...,\\
4\log(q+2)&=11.229419...,
\end{align*}
which implies the bound $\omega<3.250563$. 
The file available at \cite{File} also includes a program to check that these parameters satisfy all the constraints of Theorem~\ref{th:opt3}, as well as these calculations. 
We observe that for these parameters we have 
\begin{align*}
\ent(\Aa)&=2.14121...,\\
\ent(\Bb)&=2.14504...,\\
\sum_{(ijk)\in  \bar S_8}\alpha_{ijk}\ent(\Aa_{ijk})&=0.97706...,\\ 
\sum_{(ijk)\in \bar S_8}\alpha_{ijk}\ent(\Bb_{ijk})&=0.99206...,\\
-\chi&=0.00383...,\\
-\sum_{(ijk)\in \bar S_8}\alpha_{ijk}\chi_{ijk}&=0.01500..,
\end{align*}
so that both Inequalities \eqref{eq:CC4} and \eqref{eq:DD3} are saturated.
\section*{Acknowledgments}
This work was supported by JSPS KAKENHI grants Nos.~JP19H04066, JP20H05966, JP20H00579, JP20H04139, JP21H04879 and MEXT Quantum Leap Flagship Program (MEXT Q-LEAP) grant No.~JPMXS0120319794.

% !TEX root = ./LeG23.tex
\begin{table}[!hbp]
\begin{center}
\begin{tabular}{|c|c|}
\hline
$q$&5\\
\hline
$b$&0.9365556371\\
\hline
$\tilde b$&0.9966910575\\
\hline
$\alpha_{008}$&0.0000027575\\
$\alpha_{800}$&0.0000001000\\
$\alpha_{017}$&0.0001862683\\
$\alpha_{107}$&0.0000832431\\
$\alpha_{701}$&0.0000001000\\
$\alpha_{026}$&0.0035107863\\
$\alpha_{206}$&0.0008628502\\
$\alpha_{602}$&0.0000068293\\
$\alpha_{305}$&0.0033121296\\
$\alpha_{035}$&0.0242970779\\
$\alpha_{503}$&0.0003415907\\
$\alpha_{044}$&0.0547697273\\
$\alpha_{404}$&0.0039854191\\
$\alpha_{611}$&0.0000101383\\
$\alpha_{116}$&0.0023291693\\
$\alpha_{512}$&0.0007945380\\
$\alpha_{215}$&0.0138897271\\
$\alpha_{125}$&0.0252564389\\
$\alpha_{134}$&0.0979381496\\
$\alpha_{314}$&0.0298741641\\
$\alpha_{413}$&0.0159468865\\
$\alpha_{224}$&0.0843907598\\
$\alpha_{422}$&0.0249860330\\
$\alpha_{233}$&0.1451737658\\
$\alpha_{323}$&0.0805212325\\
\hline
$g_{0081}$&1\\
$g_{8001}$&1\\
$g_{0171}$&$1$\\
\hline
\end{tabular}
\hspace{5mm}
\begin{tabular}{|c|c|}
\hline
$g_{1071}$&$1$\\
$g_{7011}$&$1$\\
$g_{0261}$&0.3506501120\\
$g_{0262}$&0.6493498880\\
$g_{2061}$&0.3506492127\\
$g_{2062}$&0.6493507873\\
$g_{6021}$&0.3010463938\\
$g_{6022}$&0.6989536062\\
$g_{0351}$&0.0357142700\\
$g_{0352}$&0.9642857300\\
$g_{3051}$&0.0357137177\\
$g_{3052}$&0.9642862823\\
$g_{5031}$&0.0357137389\\
$g_{5032}$&0.9642862611\\
$g_{0441}$&0.0021482349\\
$g_{0442}$&0.2148227531\\
$g_{0443}$&0.7830290119\\
$g_{4041}$&0.0021482259\\
$g_{4042}$&0.2148225002\\
$g_{4043}$&0.7830292739\\
$g_{1161}$&0.2721447743\\
$g_{1162}$&0.7278552257\\
$g_{6111}$&0.1857384984\\
$g_{6112}$&0.8142615016\\
$g_{1251}$&0.0207974446\\
$g_{1252}$&0.6217716390\\
$g_{1253}$&0.3574309163\\
$g_{2151}$&0.0235531015\\
$g_{2152}$&0.6544128334\\
$g_{2153}$&0.3220340651\\
\hline
\end{tabular}
\hspace{5mm}
\begin{tabular}{|c|c|}
\hline
$g_{5121}$&0.0032282394\\
$g_{5122}$&0.6512006766\\
$g_{5123}$&0.3455710840\\
$g_{1341}$&0.0011946368\\
$g_{1342}$&0.1527219359\\
$g_{1343}$&0.0322058897\\
$g_{1344}$&0.8138775377\\
$g_{3141}$&0.0013233202\\
$g_{3142}$&0.1974477888\\
$g_{3143}$&0.0074009706\\
$g_{3144}$&0.7938279205\\
$g_{4131}$&0.0000534469\\
$g_{4132}$&0.0291489226\\
$g_{4133}$&0.0333470150\\
$g_{4134}$&0.9374506155\\
$g_{2241}$&0.0011802914\\
$g_{2242}$&0.0973553998\\
$g_{2243}$&0.2709385257\\
$g_{2244}$&0.5331703834\\
$g_{4221}$&0.0000761326\\
$g_{4222}$&0.0312465498\\
$g_{4223}$&0.3174766600\\
$g_{4224}$&0.6199541079\\
$g_{2331}$&0.0272227546\\
$g_{2332}$&0.0030572369\\
$g_{2333}$&0.3719039950\\
$g_{2334}$&0.5705932589\\
$g_{3231}$&0.0218474506\\
$g_{3232}$&0.0033067338\\
$g_{3233}$&0.3056760652\\
$g_{3234}$&0.6473222999\\
\hline
\end{tabular}
\end{center}\vspace{-2mm}
\caption{The parameters for the optimization of Section \ref{sub7} to get the upper bound $\omega(2)<3.251640$. Some parameters (e.g., $\alpha_{800}$) have value $10^{-7}$ since in our program the search space for each variable is set to the interval $[10^{-7},1]$.}\label{table:opt1}
\end{table}
% !TEX root = ./LeG23.tex
\begin{table}[!hbp]
\begin{center}
\begin{tabular}{|c|c|}
\hline
$q$&5\\
\hline
$b$&0.9369462527\\
\hline
$\tilde b$&0.9963665799\\
\hline
$\alpha_{008}$&0.0000026853\\
$\alpha_{800}$&0.0000001089\\
$\alpha_{017}$&0.0001831716\\
$\alpha_{107}$&0.0000845080\\
$\alpha_{701}$&0.0000001073\\
$\alpha_{026}$&0.0034815975\\
$\alpha_{206}$&0.0008931159\\
$\alpha_{602}$&0.0000071761\\
$\alpha_{305}$&0.0034129622\\
$\alpha_{035}$&0.0241985602\\
$\alpha_{503}$&0.0003466088\\
$\alpha_{044}$&0.0547677427\\
$\alpha_{404}$&0.0040707363\\
$\alpha_{611}$&0.0000103553\\
$\alpha_{116}$&0.0023178674\\
$\alpha_{512}$&0.0007806894\\
$\alpha_{215}$&0.0139816197\\
$\alpha_{125}$&0.0251459484\\
$\alpha_{134}$&0.0977911853\\
$\alpha_{314}$&0.0298952283\\
$\alpha_{413}$&0.0157546242\\
$\alpha_{224}$&0.0848706728\\
$\alpha_{422}$&0.0245910188\\
$\alpha_{233}$&0.1458323985\\
$\alpha_{323}$&0.0801801232\\
\hline
$g_{0081}$&1\\
$g_{8001}$&1\\
$g_{0171}$&$1$\\
\hline
\end{tabular}
\hspace{5mm}
\begin{tabular}{|c|c|}
\hline
$g_{1071}$&$1$\\
$g_{7011}$&$1$\\
$g_{0261}$&0.3506493508\\
$g_{0262}$&0.6493506492\\
$g_{2061}$&0.3387484371\\
$g_{2062}$&0.6612515629\\
$g_{6021}$&0.3901242472\\
$g_{6022}$&0.6098757528\\
$g_{0351}$&0.0357142857\\
$g_{0352}$&0.9642857143\\
$g_{3051}$&0.0409055622\\
$g_{3052}$&0.9590944378\\
$g_{5031}$&0.0459665139\\
$g_{5032}$&0.9540334861\\
$g_{0441}$&0.0021482277\\
$g_{0442}$&0.2148227712\\
$g_{0443}$&0.7830290011\\
$g_{4041}$&0.0032568914\\
$g_{4042}$&0.2797382879\\
$g_{4043}$&0.7170048207\\
$g_{1161}$&0.2703140763\\
$g_{1162}$&0.7296859237\\
$g_{6111}$&0.0000434071\\
$g_{6112}$&0.9999565929\\
$g_{1251}$&0.0203436825\\
$g_{1252}$&0.6237032466\\
$g_{1253}$&0.3559530709\\
$g_{2151}$&0.0229757598\\
$g_{2152}$&0.6447050796\\
$g_{2153}$&0.3323191606\\
\hline
\end{tabular}
\hspace{5mm}
\begin{tabular}{|c|c|}
\hline
$g_{5121}$&0.0000001000\\
$g_{5122}$&0.6524115161\\
$g_{5123}$&0.3475883839\\
$g_{1341}$&0.0011460467\\
$g_{1342}$&0.1513787481\\
$g_{1343}$&0.0315677240\\
$g_{1344}$&0.8159074812\\
$g_{3141}$&0.0012669246\\
$g_{3142}$&0.1935272338\\
$g_{3143}$&0.0087152905\\
$g_{3144}$&0.7964905510\\
$g_{4131}$&0.0000001000\\
$g_{4132}$&0.0085289607\\
$g_{4133}$&0.0327136707\\
$g_{4134}$&0.9587572686\\
$g_{2241}$&0.0011327520\\
$g_{2242}$&0.0962123683\\
$g_{2243}$&0.2731861575\\
$g_{2244}$&0.5332563541\\
$g_{4221}$&0.0000001000\\
$g_{4222}$&0.0206213157\\
$g_{4223}$&0.3269670585\\
$g_{4224}$&0.6317902100\\
$g_{2331}$&0.0268409293\\
$g_{2332}$&0.0028460012\\
$g_{2333}$&0.3782569548\\
$g_{2334}$&0.5652151854\\
$g_{3231}$&0.0219735436\\
$g_{3232}$&0.0033832115\\
$g_{3233}$&0.3042120062\\
$g_{3234}$&0.6484576951\\
\hline
\end{tabular}
\end{center}\vspace{-2mm}
\caption{The parameters for the optimization of Section \ref{sub:analysis-loss} to get the upper bound $\omega(2)<3.251502$. Some parameters (e.g., $g_{5121}$) have value $10^{-7}$ since in our program the search space for each variable is set to the interval $[10^{-7},1]$.}\label{table:opt2}
\end{table}
% !TEX root }$& ./LeG23.tex
\begin{table}[!hbp]
\begin{center}
\begin{tabular}{|c|c|}
\hline
$q$&5\\
\hline
$b$&0.9386986405\\
\hline
$\tilde b$&0.9999999000\\
\hline
$\alpha_{008}$&0.0000027522\\
$\alpha_{800}$&0.0000001000\\
$\alpha_{017}$&0.0001852694\\
$\alpha_{107}$&0.0000859878\\
$\alpha_{701}$&0.0000001000\\
$\alpha_{026}$&0.0034980486\\
$\alpha_{206}$&0.0008908045\\
$\alpha_{602}$&0.0000046150\\
$\alpha_{305}$&0.0033914676\\
$\alpha_{035}$&0.0241971433\\
$\alpha_{503}$&0.0003307825\\
$\alpha_{044}$&0.0547269864\\
$\alpha_{404}$&0.0040601320\\
$\alpha_{611}$&0.0000066778\\
$\alpha_{116}$&0.0023492364\\
$\alpha_{512}$&0.0007495489\\
$\alpha_{215}$&0.0139629817\\
$\alpha_{125}$&0.0254480643\\
$\alpha_{134}$&0.0985691144\\
$\alpha_{314}$&0.0297667390\\
$\alpha_{413}$&0.0157560121\\
$\alpha_{224}$&0.0846942718\\
$\alpha_{422}$&0.0246738020\\
$\alpha_{233}$&0.1450449800\\
$\alpha_{323}$&0.0798306553\\
\hline
$g_{0081}$&1\\
$g_{8001}$&1\\
$g_{0171}$&$1$\\
\hline
\end{tabular}
\hspace{5mm}
\begin{tabular}{|c|c|}
\hline
$g_{1071}$&$1$\\
$g_{7011}$&$1$\\
$g_{0261}$&0.3506234219\\
$g_{0262}$&0.6493765781\\
$g_{2061}$&0.3418984778\\
$g_{2062}$&0.6581015222\\
$g_{6021}$&0.2946092219\\
$g_{6022}$&0.7053907781\\
$g_{0351}$&0.0357143991\\
$g_{0352}$&0.9642856009\\
$g_{3051}$&0.0413543797\\
$g_{3052}$&0.9586456203\\
$g_{5031}$&0.0459265957\\
$g_{5032}$&0.9540734043\\
$g_{0441}$&0.0021482283\\
$g_{0442}$&0.2148218909\\
$g_{0443}$&0.7830298809\\
$g_{4041}$&0.0032308969\\
$g_{4042}$&0.2858709132\\
$g_{4043}$&0.7108981899\\
$g_{1161}$&0.2737234147\\
$g_{1162}$&0.7262765853\\
$g_{6111}$&0.7918838978\\
$g_{6112}$&0.2081161022\\
$g_{1251}$&0.0210278530\\
$g_{1252}$&0.6227553739\\
$g_{1253}$&0.3562167731\\
$g_{2151}$&0.0238860252\\
$g_{2152}$&0.6588244244\\
$g_{2153}$&0.3172895504\\
\hline
\end{tabular}
\hspace{5mm}
\begin{tabular}{|c|c|}
\hline
$g_{5121}$&0.0000001000\\
$g_{5122}$&0.6469598834\\
$g_{5123}$&0.3530400166\\
$g_{1341}$&0.0012060864\\
$g_{1342}$&0.1532284648\\
$g_{1343}$&0.0321576335\\
$g_{1344}$&0.8134078153\\
$g_{3141}$&0.0013347007\\
$g_{3142}$&0.1983635770\\
$g_{3143}$&0.0067943448\\
$g_{3144}$&0.7935073775\\
$g_{4131}$&0.0000001000\\
$g_{4132}$&0.0000001000\\
$g_{4133}$&0.0327047602\\
$g_{4134}$&0.9672950398\\
$g_{2241}$&0.0011919755\\
$g_{2242}$&0.0974575691\\
$g_{2243}$&0.2691861027\\
$g_{2244}$&0.5347067834\\
$g_{4221}$&0.0000001000\\
$g_{4222}$&0.0153803668\\
$g_{4223}$&0.3379765251\\
$g_{4224}$&0.6312626413\\
$g_{2331}$&0.0278663192\\
$g_{2332}$&0.0036110198\\
$g_{2333}$&0.3549107712\\
$g_{2334}$&0.5857455706\\
$g_{3231}$&0.0215758711\\
$g_{3232}$&0.0031002716\\
$g_{3233}$&0.3107281804\\
$g_{3234}$&0.6430198058\\
\hline
\end{tabular}
\end{center}\vspace{-2mm}
\caption{The parameters for the optimization of Section \ref{sub:rec-analysis} to get the upper bound $\omega(2)<3.250563$. Some parameters (e.g., $g_{5121}$) have value $10^{-7}$ since in our program the search space for each variable is set to the interval $[10^{-7},1]$}\label{table:opt3}
\end{table}
\newpage

%=====================================
%=====================================
\bibliographystyle{plain}
\bibliography{RMM23}
%=====================================
\end{document}